\documentclass[fontsize=11pt, headings=normal, headinclude=true, DIV=9, 
paper=a4, pagesize, cleardoublepage=current, BCOR=10mm, fleqn,
numbers=noenddot]{scrartcl}
\pdfoutput=1
\usepackage{lmodern}

\setkomafont{sectioning}{\normalcolor\bfseries}
\recalctypearea 
\setcapindent{0pt}

  \usepackage[english]{babel}  
  \usepackage[numbers,square]{natbib}
  \usepackage{mathrsfs}
  \usepackage{amsmath}
  \usepackage{amssymb}
  \usepackage{amsthm}
  \usepackage{enumitem}
  \usepackage[small,nohug,heads=LaTeX]{diagrams}
  \usepackage[utf8x]{inputenc}    
  \usepackage{lmodern}
  \usepackage[T1]{fontenc}         
  \usepackage{graphicx}            
  \usepackage[english=british]{csquotes} 

\usepackage[pdftex,plainpages=false,pdfpagelabels]
{hyperref}                                
 	\hypersetup{
 		pdftitle={The adiabatic limit of Schrödinger operators on 
fibre bundles},
  		pdfauthor={Jonas lampart and Stefan Teufel},
    	colorlinks,
  		linkcolor=black,
  		filecolor=black,
  		urlcolor=black,	
  		citecolor=black,
  		plainpages=false,
  		hypertexnames=false
}

\newtheorem{thm}{Theorem}[section]
\newtheorem{prop}[thm]{Proposition}
\newtheorem{lem}[thm]{Lemma}
\newtheorem{cor}[thm]{Corollary}

\newtheorem*{thm*}{Theorem}
\newtheorem*{Kor*}{Korollar}

\theoremstyle{definition}
\newtheorem{definition}[thm]{Definition}
\newtheorem{exm}[thm]{Example}
\newtheorem{rem}[thm]{Remark}
\newtheorem{cond}{Condition}

\newcommand{\norm}[1]{\ensuremath{\lVert #1 \rVert}}
\newcommand{\abs}[1]{\ensuremath{\lvert #1 \rvert}}
\newcommand{\eps}{\varepsilon}

\newcommand{\ud}{\mathrm{d}}
\newcommand{\ue}{\mathrm{e}}
\newcommand{\ui}{\mathrm{i}}
\newcommand{\R}{\mathbb{R}}
\newcommand{\C}{\mathbb{C}}
\newcommand{\N}{\mathbb{N}}
\newcommand{\eff}{\mathrm{eff}}

\DeclareMathOperator{\grad}{grad}
\DeclareMathOperator{\divg}{div}
\DeclareMathOperator{\dist}{dist}

\DeclareMathOperator{\tr}{tr}
\DeclareMathOperator{\supp}{supp}

\newlist{teile}{enumerate}{2}
\setlist[teile,1]{label=\arabic*),fullwidth}
\setlist[teile,2]{label=(\alph*), fullwidth}

\begin{document}
\title{The adiabatic limit of Schr\"odinger operators on fibre bundles}
\author{Jonas Lampart \thanks{PSL Research University \& CEREMADE (UMR CNRS 7534), Universit\'e  de Paris-Dauphine, Place du Marchal de Lattre de Tassigny, 75775 Paris Cedex 16, France. \texttt{lampart@ceremade.dauphine.fr}},
Stefan Teufel \thanks{Fachbereich Mathematik, Universit\"at T\"ubingen, Auf der Morgenstelle 10, 72076 T\"ubingen, Germany.
\texttt{stefan.teufel@uni-tuebingen.de}}}

\maketitle
\begin{abstract}
We consider Schrödinger operators $H=-\Delta_{g_\eps} + V$ on a fibre bundle $M\stackrel{\pi}{\to}B$ with compact fibres and a metric $g_\eps$ that blows up directions perpendicular to the fibres by a factor ${\eps^{-1}\gg 1}$. We show that for an eigenvalue $\lambda$ of the fibre-wise part of $H$, satisfying a local gap condition, and every $N\in \N$ there exists a subspace of $L^2(M)$
that is invariant under $H$ up to errors of order $\eps^{N+1}$. The dynamical and spectral features of $H$ on this subspace can be described by an effective operator on the fibre-wise $\lambda$-eigenspace bundle $\mathcal{E}\to B$, giving detailed asymptotics for $H$.
\end{abstract}
\tableofcontents
\section{Introduction}
The adiabatic limit of a fibre bundle, in which lengths in the fibres are of size $\eps \ll 1$ compared to those in the base, is the stage of many interesting results on the Laplacian of a Riemannian manifold.
For instance the study of the Schrödinger equation~\cite{deO, deOVe} or the heat equation~\cite{Kreko, Wi} in thin tubes, with Dirichlet boundary conditions, reveals a variety of interesting effects. These are also found in many works on the relation of spectrum and geometry in these special spaces (see~\cite{BMT, BGRSweakly, CEKtop, CDFK, DEbound, FrSo,  Ga, GoJa, Gru, KoVu, KrLu, LiLu2} and references therein). 
A related question is the confinement of a system to a submanifold of its configuration space by a scaled family of potentials~\cite{daC, FrHe, Mar, Mit, WaTeConst} that effectively force the system into a thin tubular neighbourhood of the submanifold of the type studied in the works above. Tubular neighbourhoods of graphs and boundary conditions other than Dirichlet are discussed in detail in~\cite{Gr, Po}. The discreteness of the entire spectrum on fibre bundles with closed fibres was investigated in~\cite{Bai, BMP, Bor, Kle}.
Kordyukov considers foliated manifolds~\cite{Kord} and other authors study the Hodge Laplacian~\cite{LoKo,Lo, MaMe} or Dirac operators~\cite{BiCh, Dai, Goe} in the adiabatic limit. 

The aim of this paper is to study the structure behind many of the specific results above and the various approximation techniques used to prove them.
We will develop a general method, inspired by ideas that were originally introduced in the analysis of magnetic Schr\"odinger operators~\cite{HeSj1, HeSj2, PSTpeierl} and the Born-Oppenheimer approximation~\cite{MaSo, PSTspace}.
These will be cast into a new form suited to our geometric context.
We identify natural conditions on the geometry and the operator in question for the validity of such approximations and refine them by deriving expansions to arbitrary powers of $\eps$. Our results will lead to generalizations of many specific results in the literature, although here we will focus on the general reasoning behind these results, rather than trying to emulate them in full technical detail.
We show how the general approximation scheme we derive here can be used to expand and unify the large literature on thin tubes around submanifolds, often called \enquote{quantum waveguides}, in another work in collaboration with Haag~\cite{HLT}. 
There we also examine several new examples.
The strength of our method also allows for the study of nodal sets of eigenfunctions and their limits. 
The limiting behaviour of these sets is studied by the first author in~\cite{LaNod, Lam}. This addresses various questions on the nodal count and the relation of the nodal set to the boundary, studied also in~\cite{FrKr, GrJe, JerDiam, Jer,  KreTu}.
For small, simple eigenvalues conditions are found under which the nodal set of an eigenfunction must intersect the boundary and for $B=S^1$ this set shown to be isotopic to a disjoint union of fibres.

Better understanding of the underlying adiabatic structure should prove fruitful also for those problems we do not explicitly treat here. In particular a generalisation of our method to the Hodge Laplacian on differential forms is rather natural.
We will discuss the related literature in more detail at the end of Section~\ref{sect:thm}, after stating our main results.

Let $M\stackrel{\pi}{\to}B$ be a fibre bundle of smooth manifolds with boundary. We assume the fibre $F$ to be compact, with or without boundary, and the base $B$ to be complete, so in particular $\partial B=\varnothing$, but in general not compact.
Denote by $TF:=\ker \pi_*$ the vertical subbundle of $TM$. Let $g$ and $g_B$ be Riemannian metrics on $M$ and $B$ such that $\pi_*$ induces an isometry $TM/TF\to TB$. Then $g$ is called a Riemannian submersion and can be written in the form
\begin{equation*}
g=g_F+\pi^*g_B\,,
\end{equation*}
where $g_F$ vanishes on the horizontal subbundle $NF:=TF^\perp$. We will require that $(M,g)\to (B,g_B)$ be a fibre bundle of bounded geometry (see Section~\ref{sect:geom} for a precise definition).
The family of metrics 
\begin{equation*}
 g_\eps:=g_F+ \eps^{-2} \pi^*g_B
\end{equation*}
for $\eps\ll 1$ is called the adiabatic limit of $(M,g)$. In this limit we consider a Schrödinger operator of the form
\begin{equation}\label{eq:H_intro}
 H:=-\Delta_{g_\eps} + V + \eps H_1\,,
\end{equation}
where $V$ is a potential and $H_1$ is a second order differential operator, that may for instance model a perturbation of the metric $g_\eps$. Such perturbations arise naturally if we think of $M$ as being embedded in a, suitably rescaled, thin tubular neighbourhood of $B$, which is embedded as a submanifold in $\R^k$ or some Riemannian manifold (see~\cite{HLT}).

With Dirichlet conditions on $\partial M$ this operator is self-adjoint on
$D(H)\subset \mathscr{H}:=L^2(M,g)$
(the precise technical conditions on $H$ will be stated in Section~\ref{sect:proj}). The Laplacian of $g_\eps$ decomposes with respect to horizontal and vertical directions as
\begin{equation*}
\Delta_{g_\eps}=\Delta_F + \eps^2\Delta_h\,,
\end{equation*}
where the fibre-wise action of $\Delta_F$ is that of the Laplacian of the vertical metric $g_F$ and
\begin{equation*}
 \Delta_h= \tr_{NF} \nabla^2 - \eta\,,
\end{equation*}
with the Levi-Cività connection $\nabla$ and the mean curvature vector $\eta$ of the fibres of $(M,g)$ (for $\eps=1$). Hence we can write
\begin{equation*}
 H=-\eps^2 \Delta_h + \eps H_1 + H_F\,,
\end{equation*}
with the fibre-wise operator
\begin{equation*}
 H_F:=-\Delta_F + V\,.
\end{equation*}
Consider the restriction $H_{F}(x)$ of this operator to the fibre $F_x$ with Dirichlet conditions on $\partial F_x$, i.e.~on the domain $W^{2,2}(F_x)\cap W^{1,2}_0(F_x)$ (since we work exclusively in $L^2$ we will drop the corresponding superscript from now on and write $W^{k,2}=W^k$ etc.).
This operator is self-adjoint and its spectrum consists of real eigenvalues of finite multiplicity accumulating at infinity. An \textit{eigenband} of $H_F$ is a function ${\lambda{:}\,B\to \R}$ with $\lambda(x)\in \sigma(H_F(x))$ for every $x\in B$. For any such eigenband and $x\in B$ let $P_0(x)$ be the orthogonal projection to $\ker(H_F(x) -\lambda(x))$ in $L^2(F_x)$.
This projection $P_0$ is the starting point for our analysis. We will adopt two complementary points of view of this operator and similar objects. The first is to view $P_0$ as an operator on functions defined on $M$, whose image consists exactly of those functions whose restrictions to any fibre $F_x$ are $\lambda(x)$-eigenfunctions of $H_F(x)$. The other view is that such fibre-wise operators are sections of certain vector bundles of infinite rank over $B$ induced by the fibre bundle structure, as we will now explain.

Let $\mathscr{H}_F$ be the vector bundle over $B$, with fibre $L^2(F)$, defined by the transition functions
\begin{equation*}
T_\Phi{:}\, U \times L^2(F) \to U \times  L^2(F)\,,\qquad f\mapsto f\circ \Phi\,,
\end{equation*}
where $\Phi{:}\, U \times F \to U \times F$, $\Phi(x,y)=(x,\phi_x(y))$ is a transition function between different trivialisations of $\pi^{-1}(U)\subset M$. The fibre of this vector bundle is defined by the topological vector space $L^2(F)$, which is isomorphic to the Hilbert space $L^2(F_x, g_{F_x})$ for every $x$. Let $D(H_F)$ be the vector bundle with fibre $W^2(F)\cap W^1_0(F)$ constructed in the same way. The latter implements Dirichlet conditions for $H_F$ on $\partial M$.

These bundles are hermitian vector bundles with the natural fibre-wise pairings induced by the metric $g_{F}$. The spaces of continuous fibre-wise maps between vector bundles naturally have a bundle structure, and we can immediately observe that
\begin{equation*}
 H_F \in L^\infty\big(\mathscr{L}(D(H_F), \mathscr{H}_F)\big)\qquad \text{and} \qquad 
 P_0\in L^\infty(\mathscr{L}\big(\mathscr{H}_F)\big)
\end{equation*}
are bounded sections of these bundles. Now if $P_0$ is a continuous section of this bundle, $\mathrm{rank}\, P_0=\tr P_0$ is continuous, whence it is constant and ${\mathcal{E}:=P_0 \mathscr{H}_F}\subset \mathscr{H}_F$ is a subbundle of finite rank. Its fibre over $x$ is exactly $\ker(H_F(x) -\lambda(x))$. 

Since the space of $L^2$-sections $L^2(\mathscr{H}_F)$ is isomorphic to $\mathscr{H}$ (cf.~\cite[Appendix B]{Lam}), we can also view $P_0$ as a bounded linear map on $\mathscr{H}$. The image of this map then consists exactly of the $L^2$-sections of the $\lambda$-eigenspace bundle $L^2(\mathcal{E})\cong P_0 \mathscr{H}$. This gives a precise meaning to $P_0$ as an operator on functions on $M$ and we will now use both of these views alongside each other, without distinguishing them by the notation.

We will consider eigenbands $\lambda$ that have a \textit{spectral gap}:
\begin{cond}\label{cond:gap}
 There exist $\delta>0$ and bounded continuous functions $f_-,f_+\in \mathscr{C}_b(B)$ with $\dist\left(f_{\pm}(x),\sigma(H_F(x))\right)\geq \delta$ such that 
\begin{equation*}
 \forall x\in B: [f_-(x),f_+(x)]\cap \sigma(H_F(x))=\lambda(x)\,.
\end{equation*}
\end{cond}
If $F$ is connected,  this condition  is always satisfied for the ground state band $\lambda_0(x):=\min \sigma(H_F(x))$ (see Proposition~\ref{prop:gap}). We note that all previous works in our list of references with the exception of~\cite{WaTeConst} were solely concerned with the ground state band and energies very close to its global minimum. Moreover, it is well known that in regions of $B$ near crossings of different eigenbands the adiabatic approximation breaks down (see e.g. Fermanian-Kammerer  and G\'erard~\cite{FeGe}). Hence, the gap condition is not a purely technical restriction, but a necessary ingredient for adiabatic decoupling. 

Later we will prove that Condition~\ref{cond:gap} implies continuity of $P_0$. In~\cite{Lam} it is shown that $\mathcal{E}$ has a smooth structure such that $\Gamma(\mathcal{E})\subset \mathscr{C}^\infty(M, \C)$, so we can think of any smooth section of $\mathcal{E}$ as a smooth function $\phi$ on $M$ satisfying $H_F(x)\phi=\lambda(x)\phi$ on every fibre $F_x$. Since this condition is independent of $\eps$ we can define an $\eps$-dependent family of sections as  a product $\phi\pi^*\psi$, where  only $\psi \in  \mathscr{C}^\infty(B)$ depends on $\eps$. Any section of $\mathcal{E}$ may be written as a sum of such products involving a finite number of generators $\phi$ (of $\Gamma(\mathcal{E})$ over $\mathscr{C}^\infty(B)$).

The adiabatic approximation with respect to such an eigenband $\lambda$ consists in projecting $H$ with $P_0$. This approximation is good if we are concerned with states in (or close to) the image of $P_0$ and this space is approximately invariant under $H$. Since
\begin{equation*}
 (H- P_0 H P_0 )P_0=[H,P_0]P_0
\end{equation*}
this basically means that the commutator $[H,P_0]$ needs to be small. We can see, at least heuristically, that $P_0 \mathscr{H}$ is invariant under $H$ up to errors of order $\eps$.
Let $\phi\pi^*\psi$ be a section of $\mathcal{E}$  with $\phi$ independent of $\eps$ as described above.
Now if $X$ is a horizontal vector, then $\eps X$ has length $\mathcal{O}(1)$ in $(M,g_\eps)$  but $\eps X \phi$ is of order $\eps$. Hence we have
\begin{equation*}
 [\eps X, P_0]\phi \pi^*\psi = \eps (1-P_0)X(\phi \pi^*\psi)=(\pi^*\psi) \eps(1-P_0)X\phi=\mathcal{O}(\eps)\,.
\end{equation*}
The commutator $[H,P_0]$ can be expressed by such derivatives, so it is also of order $\eps$. This is an instance of the more general principle that horizontal derivatives of quantities associated with the fibres, which are independent of $\eps$, should be small on $(M,g_\eps)$ as $\eps\to 0$. We will use this intuition to construct a projection $P_\eps$ with $[P_\eps, H]=\mathcal{O}(\eps^{N+1})$ by 
eliminating commutators order by order in $\eps$.

The image of $P_\eps$ is then almost-invariant under $H$. Spaces with this property were constructed for different problems in~\cite{HeSj1, HeSj2, MaSo,PSTpeierl, PSTspace}, using pseudo-differential calculus. Our method is based on similar ideas, but we develop a new technical framework to implement them, for two reasons. First, the required calculus of operator-valued pseudo-differential operators on manifolds, that must allow for complete symbol expansions, is not well established to date. Second, the presence of the boundary poses an additional difficulty, that is more easily controlled using special classes of differential operators (see section~\ref{sect:proj}).

\section{Main results}\label{sect:thm}
Before going into the details of the construction we state our main results and derive some corollaries. All of the statements here require that $(M,g)\to (B,g_B)$ is of bounded geometry (see Condition~\ref{cond:geom}) and that $H$, given by~\eqref{eq:H_intro} with domain $D(H)=W^2(M,g)\cap W^1_0(M,g)$, satisfies Condition~\ref{cond:H}. If not mentioned otherwise, they hold for any eigenband $\lambda$ with a spectral gap (Condition~\ref{cond:gap}), so if $H_F$ has multiple such bands the theorems can be applied to the same operator in a variety of ways, yielding for example results for different energies. If all of the eigenvalues of $H_F$ belong to such bands, we can in principle apply the construction to every one of them and obtain a total decomposition of $H$ into operators $H_j$, labeled by the different bands. 
\begin{thm}\label{thm:proj}
 For every $\Lambda>0$ and $N\in \N$ there exists an orthogonal projection $P_\eps \in \mathscr{L}(\mathscr{H})\cap \mathscr{L}(D(H))$ that satisfies
\begin{equation*}
 \norm{\left[H, P_\eps\right] \varrho(H)}_{\mathscr{L}\left(\mathscr{H}\right)} = \mathcal{O}(\eps^{N+1})
\end{equation*}
for every Borel function $\varrho{:}\,\mathbb{R}\to [0,1]$ with support in $(-\infty, \Lambda]$.
Furthermore $P_\eps - P_0 =\mathcal{O}(\eps)$ in $\mathscr{L}(D(H))$ and there exists a unitary $U_\eps$ on $\mathscr{H}$ that maps $P_0\mathscr{H}=L^2(\mathcal{E})$ to the image of $P_\eps$.
\end{thm}
The step-by-step construction of this projection will make up a large part of Section~\ref{sect:proj}. Once this is achieved, the unitary may be defined by the Sz.-Nagy formula
\begin{equation}\label{eq:Udef}
 U_\eps:=\big( P_\eps P_0 + (1-P_\eps)(1-P_0)\big)\big(1-(P_0-P_\eps)^2\big)^{-1/2}\,,
\end{equation}
since $1-(P_0-P_\eps)^2$ is positive for $\eps$ small enough.

Though this theorem seems rather technical, it is convenient for the derivation of statements on the dynamical and spectral properties of $H$. Firstly one can show that the image of $P_\eps$ is almost-invariant under $\ue^{-\ui Ht}$ using standard time-dependent perturbation theory (see e.g. Wachsmuth and Teufel~\cite[Section 3.1]{WaTeConst}).
\begin{cor}\label{cor:inv}
For every $N\in \N$ and $\Lambda>0$ there exist constants $C$ and $\eps_0>0$ such that 
if $P_\eps$ and $\varrho$ are as in Theorem~\ref{thm:proj} we have
\begin{equation*}
\norm{[\ue^{-\ui Ht}, P_\eps]\varrho(H)}_{\mathscr{L}(\mathscr{H})}\leq C \eps^{N+1}\abs{t}
\end{equation*}
for all $\eps\in(0, \eps_0)$ and $t\in \R$
\end{cor}
It is of high interest to provide an effective description of $H$ on $L^2(\mathcal{E})$, for which this dynamical invariance is a necessary requirement. Such a description is provided by the effective operator
\begin{equation*}
H_\mathrm{eff}:=U_\eps^* P_\eps H P_\eps U_\eps\,.
\end{equation*}
Whenever $\mathcal{E}$ is a trivial line-bundle, in particular for the ground state band $\lambda_0$, this is an operator in $L^2(B)$. As such it provides a description of $H$ on $M$ by an operator on the lower dimensional space $B$.
Since $[P_\eps, H]=[P_0, H] + \mathcal{O}(\eps)=\mathcal{O}(\eps)$, one easily checks that $H_\mathrm{eff}$ is self-adjoint on $D_\mathrm{eff}=U_\eps^*P_\eps D(H)\subset L^2(\mathcal{E})$ using the Kato-Rellich theorem.
Because of the invariance of $P_\eps \mathscr{H}$, the solutions to the Schrödinger equation with initial data in the image of $P_\eps$ can be approximated using
this operator.
\begin{cor}
Let $\Lambda>0$, $N\in \N$ and $P_\eps$ be the corresponding projection of Theorem~\ref{thm:proj}. There exist constants $C$ and $\eps_0>0$ such that
\begin{equation*}
 \big\lVert\big(\ue^{-\ui Ht}-U_\eps \ue^{-\ui H_\mathrm{eff} t}U_\eps^*\big) P_\eps 1_{(-\infty,\Lambda]}(H)\big\rVert_{\mathscr{L}(\mathscr{H})}\leq C \eps^{N+1} \abs{t}
\end{equation*}
for every $\eps\in(0, \eps_0)$ and $t\in \R$.
\end{cor}
\begin{proof}
Application of Duhamel's formula gives
\begin{equation*}
 \ue^{-\ui Ht}-U_\eps \ue^{-\ui H_\mathrm{eff} t}U_\eps^*
= -\ui\int_0^t U_\eps \ue^{-\ui H_\mathrm{eff} (t-s)}U_\eps^*
(H-U_\eps H_\mathrm{eff} U_\eps^*)
\ue^{-\ui Hs}\,\ud s\,
\end{equation*}
and since $U_\eps H_\eff U^*_\eps =P_\eps H P_\eps$ commutes with $P_\eps$ we have
\begin{align*}
\big(\ue^{-\ui Ht}&-U_\eps \ue^{-\ui H_\mathrm{eff} t}U_\eps^*\big) P_\eps 1_{(-\infty,\Lambda]}(H)\\
&=\Big(P_\eps \big(\ue^{-\ui Ht}-U_\eps \ue^{-\ui H_\mathrm{eff} t}U_\eps^*\big)
+ [\ue^{-\ui Ht},P_\eps]\Big) 1_{(-\infty,\Lambda]}\\
&=\begin{aligned}[t]
\Big(-\ui P_\eps\int_0^t U_\eps \ue^{-\ui H_\mathrm{eff} (t-s)}U_\eps^*(H-P_\eps H P_\eps&)\ue^{-\ui Hs}\,\ud s\\
+&[\ue^{-\ui Ht},P_\eps]\Big)1_{(-\infty,\Lambda]}
\end{aligned}\\
&=\begin{aligned}[t]
-\ui \int_0^t U_\eps \ue^{-\ui H_\mathrm{eff} (t-s)}U_\eps^* 
\underbrace{(P_\eps H- P_\eps H P_\eps)}_{=-P_\eps[H,P_\eps]}&1_{(-\infty,\Lambda]}\ue^{-\ui Hs}\,\ud s \\
&+[\ue^{-\ui Ht},P_\eps]1_{(-\infty,\Lambda]}.
\end{aligned}
\end{align*}
This is of order $\eps^{N+1}\abs{t}$, since the integrand of the first term is of order $\eps^{N+1}$ by Theorem~\ref{thm:proj} with $\varrho=1_{(-\infty,\Lambda]}$, and the second term is of this order by Corollary~\ref{cor:inv}.
\end{proof}
Clearly such techniques can also be used to derive properties of the heat semigroup. A first result on the spectrum of $H$ is also obtained in a straightforward manner.
\begin{cor}\label{cor:spectrum}
For arbitrary but fixed $\Lambda>0$ and $N\in \N$ let $H_\mathrm{eff}$ be the corresponding effective operator. 
Then for every $\delta>0$ there exist constants C and $\eps_0>0$ such that for every ${\mu\in \sigma(H_\mathrm{eff})}$ with ${\mu\leq \Lambda-\delta}$:
\begin{equation*}
 \mathrm{dist}(\mu, \sigma(H))\leq C \eps^{N+1}
\end{equation*}
for every $\eps\in(0, \eps_0)$.
\end{cor}
\begin{proof}
 Let $(\psi_k)_{k\in\mathbb{N}}\subset L^2(\mathcal{E})$ be a Weyl sequence for $\mu$ (i.e.~$\norm{\psi_k}=1$ for every $k\in\mathbb{N}$ and $\lim_{k\to \infty} \norm{\left(H_\mathrm{eff}-\mu\right)\psi_k}=0$). We can even choose the $\psi_k$ in the image of $1_{(-\infty,D]}(H_\mathrm{eff})$, with $D=\Lambda -\delta/2$, because $\mu$ is in the spectrum of $H_\eff$ restricted to this space.
Then because $\psi_k\in P_0 \mathscr{H}$ we have
\begin{align}
 \lVert (&H-\mu)U_\eps \psi_k \rVert_\mathscr{H}\notag\\
& =\norm{\left(H-\mu\right)P_\eps U_\eps 1_{(-\infty,D]}(H_{\mathrm{eff}}) \psi_k}\notag\\
&\leq \norm{U_\eps\left(H_\mathrm{eff}-\mu\right)\psi_k}
+\lVert P_\eps^\perp HP_\eps U_\eps 1_{(-\infty,D]}(H_\mathrm{eff}) \psi_k\rVert\,.\label{eq:quasimode0}
\end{align}
Let $\chi\in \mathscr{C}_0^\infty\big((-\infty, \Lambda], [0,1]\big)$ be equal to one on a set containing $\sigma(H_\mathrm{eff})\cap(-\infty, D]$ (and sufficiently regular, see Definition~\ref{def:cutoff}). Then by the functional calculus (cf.~Lemma~\ref{lem:chi})
\begin{align*}
 U_\eps 1_{(-\infty,D]}(H_\mathrm{eff})&=\chi(P_\eps H P_\eps)U_\eps 1_{(-\infty,D]}(H_\mathrm{eff})\\
 &=P_\eps \chi(H)P_\eps U_\eps 1_{(-\infty,D]}(H_\mathrm{eff})+ \mathcal{O}(\eps^{N+1})\\
&=\chi(H)P_\eps U_\eps 1_{(-\infty,D]}(H_\mathrm{eff})+ \mathcal{O}(\eps^{N+1})
\,.
\end{align*}
Hence Theorem~\ref{thm:proj} with $\varrho=\chi$ gives a bound on the second term
\begin{align*}
P_\eps^\perp H P_\eps U_\eps 1_{(-\infty,D]}(H_\mathrm{eff})
&= [H,P_\eps]\chi(H) P_\eps U_\eps 1_{(-\infty,D]}(H_\mathrm{eff})+ \mathcal{O}(\eps^{N+1})
\\&=\mathcal{O}(\eps^{N+1})\,.
\end{align*}
For the first term we can then simply choose $k$ large enough for it to be smaller than the second one. This shows that for $\varphi=U_\eps \psi_k$
\begin{equation*}
 \lVert (H-\mu)\varphi \rVert_\mathscr{H} \leq C \eps^{N+1}\,.
\end{equation*}
So either $(H-\mu)\varphi=0$ and $\mu$ is an eigenvalue of $H$, or the vector $\norm{(H-\mu)\varphi}^{-1}{(H-\mu)\varphi}$ is normalised and
\begin{align*}
\mathrm{dist}(\mu, \sigma(H))^{-1}
&=\lVert\left(H-\mu\right)^{-1}\rVert_{\mathscr{L}(\mathscr{H})}\\
&\geq \frac{1}{\lVert (H-\mu)\varphi \rVert} \lVert (H-\mu)^{-1}(H-\mu) \varphi \rVert_\mathscr{H}\\
&\geq\frac{1}{C \eps^{N+1}}\,.
\end{align*}
\end{proof}
In this proof we used the functions $U_\eps \psi_k$ as quasi-modes for $H$. If 
we have $\mu\in \sigma(H)$ with a Weyl sequence $(\varphi_k)_{k\in\N}$ the 
natural choice of quasi-modes for $H_\mathrm{eff}$ is $U_\eps^*P_\eps 
\varphi_k$. If the norm of this sequence is bounded below, we can easily 
reproduce the proof to obtain $\dist(\mu, 
\sigma(H_\mathrm{eff}))=\mathcal{O}(\eps^{N+1})$. Of course if $\mu$ is 
associated with a different eigenband $\tilde \lambda$ than the one we used for 
the construction of $P_\eps$ this will not be the case. If however $\Lambda$ is 
small, the only contribution should be that of the ground state band 
$\lambda_0(x)$. 
To be more precise, 
let $\lambda_1(x):= \min \big(\sigma(H_F(x))\setminus\lbrace\lambda_0\rbrace \big)$ be the second eigenvalue of $H_F(x)$ and put $\Lambda_1:= \inf_{x\in B}\lambda_1(x)$. Let 
$P_\eps$ be the 
super-adiabatic projection constructed for $\lambda_0$ and some $\Lambda$ and $N$.
Then the quadratic form of $H$ on $P_\eps^\perp \mathscr{H}$ satisfies, for every normalised $\psi\in P_\eps^\perp D(H)$,
\begin{align}
 \langle \psi, P_\eps^\perp H P_\eps^\perp \psi\rangle&=
 \langle P_0^\perp \psi,H P_0^\perp \psi \rangle + \mathcal{O}(\eps)\label{eq:lowerbound}\\
 &=\langle P_0^\perp \psi, (-\eps^2 \Delta_h + \eps H_1) P_0^\perp \psi \rangle
 + \underbrace{\langle P_0^\perp \psi, H_F P_0^\perp \psi \rangle}_{\geq \Lambda_1} + \mathcal{O}(\eps)\,.
 \notag
\end{align}
Now if $-\eps^2 \Delta_h + \eps H_1$ is positive (or bounded below by a term of order $\eps$), then the operator $P_\eps^\perp H P_\eps^\perp$ is bounded below by $\Lambda_1 - \mathcal{O}(\eps)$ and thus has no spectrum below this threshold. In a sense, this implies that eigenvalues of $H$ below $\Lambda_1$ must be associated with the effective operator of the ground state band $\lambda_0$. This can be formalised in the following way:
\begin{prop}\label{prop:ground}
Assume $F$ is connected and that $-\eps^2\Delta_h + \eps H_1$ is bounded below by $-C\eps$.
Let $H_\mathrm{eff}$ be the effective operator for $\lambda_0$ and some $\Lambda>0$ and $N\in \N$. Then for every regular cut-off $\chi \in \mathscr{C}^\infty_0\big((-\infty, \Lambda_1), [0,1]\big)$ (cf.~Definition~\ref{def:cutoff}) we have
\begin{equation*}
\norm{H \chi(H) - U_\eps H_\mathrm{eff}  \chi(H_\mathrm{eff})U_\eps^*}=\mathcal{O}(\eps^{N+1})\,.
\end{equation*}
In particular for every $\delta>0$
\begin{equation*}
\dist\big( \sigma(H) \cap(-\infty, \Lambda_1 -\delta], \sigma(H_\mathrm{eff}) \cap(-\infty, \Lambda_1 -\delta]\big)=\mathcal{O}(\eps^{N+1})\,,
\end{equation*}
where $\dist$ denotes the Hausdorff distance between compact subsets of $\R$.
\end{prop}
The projection $P_\eps$ has an explicit recursive construction that allows for an expansion of the effective operator. This expansion involves differential operators of increasing order, so uniform estimates can only hold after cutting off high energies. Up to order $\eps^3$ this expansion consists of the adiabatic operator
\begin{equation*}
 H_\mathrm{a}:=P_0HP_0= -\eps^2 P_0\Delta_h P_0 + \lambda +  \eps P_0 H_1 P_0
\end{equation*}
and the first super-adiabatic correction
\begin{equation*}
 \mathcal{M}:= P_0 [H,P_0]R_F(\lambda)[H,P_0]P_0\,,
\end{equation*}
where $R_F(\lambda)$ denotes the reduced resolvent $(H_F-\lambda)^{-1}P_0^\perp$ of the eigenband $\lambda$.
The expansion is accurate to order $\eps^3$ in the sense that
\begin{equation*}
 \norm{H_\mathrm{eff}\chi^2(H_\mathrm{eff}) - \chi(H_\mathrm{eff})\left(H_\mathrm{a} + \mathcal{M}\right) \chi(H_\mathrm{eff})}_{\mathscr{L}(\mathscr{H})}=\mathcal{O}(\eps^3)\,,
\end{equation*}
for a sufficiently regular function $\chi\in \mathscr{C}^\infty_0\big((-\infty, \Lambda], [0,1]\big)$.
A proof of this statement can be found in~\cite[Section 2.2.1]{Lam}. Here we will not give the details of this proof, but the form of the expansion will become clear from the construction of $P_\eps$.

We may observe here that $\mathcal{M}$ is a fourth-order differential operator, which is the reason why energy cut-offs were needed in the statements of all the previous theorems. This also suggests that at small energies the super-adiabatic corrections might be of less importance and the adiabatic operator $H_\mathrm{a}$ already provides a good description. This is true under some additional assumptions on $H_1$ (see Section~\ref{sect:low}, Condition~\ref{cond:low}). Our main results on the spectrum at energies of order $\eps^\alpha$ above $\Lambda_0:=\inf_{x\in B} \lambda_0(x)$ are (the precise statements are given in Section~\ref{sect:low}):
\begin{enumerate}
 \item 
 $\dist\big(\sigma(H)\cap (-\infty, \Lambda_0 + C\eps^\alpha],\sigma(H_\mathrm{a})\cap (-\infty,  \Lambda_0 + C\eps^\alpha]\big)=\mathcal{O}(\eps^{2+\alpha/2})$.
 \item If $H_\mathrm{a}$ has $K+1$ eigenvalues $\mu_0< \mu_1\leq \dots\leq \mu_K 
< C\eps^\alpha$ below the essential spectrum, then the bottom of $\sigma(H)$ 
also consists of eigenvalues $\lambda_0<\lambda_1\leq  \dots\leq \lambda_K$ and 
for $0\leq j\leq K$:
 \begin{equation*}
  \abs{\mu_j - \lambda_j}=\mathcal{O}(\eps^{2+\alpha})\,.
 \end{equation*}
\end{enumerate}
%
%
The result on low lying eigenvalues is relevant also if $B$ is not compact, as it has been stressed in the literature~\cite{BGRSweakly, CEKtop, CDFK, DEbound, GoJa, Gru, LiLu2} that such eigenvalues exist in interesting applications. In Section~\ref{sect:low} we will also derive some results on the approximation of eigenfunctions. These are relevant for the application to nodal sets~\cite{LaNod} where it is shown that, for certain low lying eigenvalues, the behaviour of the nodal set is essentially determined by the nodal set of the eigenfunction of $H_\mathrm{a}$ with the corresponding eigenvalue.

A large portion of the literature on the adiabatic limit of Schrödinger operators is concerned with quantum waveguides. There, one starts with some sort of tubular neighbourhood of an embedded submanifold of $\R^m$. This leads to a fibre bundle $M$ with base diffeomorphic to that submanifold, as we describe in detail in~\cite{HLT} (see also~\cite[Chapter 3]{Lam}). 

The most commonly treated case is $B=\R$ or $B=I\subset \R$ an interval~\cite{BMT,BGRSweakly,CDFK,deO,deOVe,DEbound,FrSo,Ga,GoJa}. The fibre of such a tube is a compact domain whose dimension is the codimension of $B$. Topologically $M$ is the product of a finite or infinite interval and a compact domain. This seemingly simple situation already allows for several different effects that depend on the codimension of $B$ and manifest themselves in the metric of $M$, both in the choice of horizontal bundle $NF$ (see Example~\ref{ex:vfields}) and as corrections to the submersion metric $g_\eps$ that we treat in form of the perturbation $H_1$.
The authors of~\cite{BGRSweakly, FrSo} treat a tubular neighbourhood of varying width of the $x$-axis in $\R^2$. In~\cite{CDFK, DEbound, GoJa} the emphasis is put on the effect of \enquote{bending}, i.e.~the exterior curvature of the submanifold. The similar case with an embedding into a complete, non-compact surface $\Sigma$, in which also the curvature of $\Sigma$ plays a role, was treated by Krej{\v{c}}i{\v{r}}{\'\i}k~\cite{Krestrip}. Additionally the effect of \enquote{twisting}, which is present when $NF$ is not tangent to $\partial M$, is discussed in~\cite{BMT, deO, deOVe, Ga}. 

The results for bases of higher dimensions are far less detailed. In~\cite{CEKtop} and earlier works Carron, Exner and Krej{\v{c}}i{\v{r}}{\'\i}k study embeddings of surfaces into $\R^3$, while Lin and Lu~\cite{LiLu2} consider special submanifolds of $\R^k$ of arbitrary dimension and codimension. The induced metrics on the base are assumed to be geodesically complete and asymptotically flat, and $M$ is taken as a neighbourhood of zero in $NB$, whose fibre is a metric ball of fixed radius. Wittich~\cite{Wi} treats tubular neighbourhoods of compact manifolds in a Riemannian manifold $(A,g_A)$ whose fibres are geodesic balls in the normal directions. The emphasis of these works are the effects of extrinsic and intrinsic curvature on the spectrum and the resolvent of the Laplacian.

More general manifolds have been considered with metrics that are of a simpler structure than those arising from embeddings. In this context one is usually concerned with closed fibres. Baider~\cite{Bai} works with warped products, Kleine~\cite{Kle} treats more general metrics on manifolds of the form $\R_+ \times F$ and the authors of~\cite{BMP, Bor} study Riemannian submersions with some additional assumptions on the mean curvature vector of the fibres. The works~\cite{Bai, Bor, BMP, Kle} derive conditions for the Laplacian on a non-compact $M$ to have discrete spectrum. These will not be satisfied under our technical assumptions (Section~\ref{sect:geom}), since they require non-uniform behaviour of the geometry. But although our uniform estimates do not hold for these spaces, the effective operators can be formally calculated and give good intuition both for understanding such results and for possible choices of trial functions to prove them.

In~\cite{FrHe, WaTeConst} localisation to submanifolds is achieved using potentials rather than boundary conditions. The localisation is used to reformulate the problem on the normal bundle $NB$, so the structure is very similar to our problem, with $M=NB$ and a potential $V\neq 0$ of a form that gives localised eigenfunctions of $H_F=-\Delta_F +V$. Froese and Herbst~\cite{FrHe}  assume that $B$ is a compact, complete submanifold of $\R^m$, while in~\cite{WaTeConst} the base and the ambient space in which it is embedded are, apart from technical assumptions, basically arbitrary complete Riemannian manifolds. The leading order of the metric on $NB$ arising in this situation is the Sasaki metric, which is a Riemannian submersion with totally geodesic fibres.

Our approach considerably generalises the geometries that have been considered in the literature. On the one hand we consider very general fibre bundles without reference to an embedding, while on the other hand we include the flexibility needed to deal with the complicated metrics such an embedding may induce. This shows that a large class of problems have the sufficient structure for adiabatic techniques to be applicable. Our results also complement the previously studied quantum waveguides by allowing for generic deformations of the fibres, as opposed to scaling and twisting only. For example, one may think of deforming a disk into an elliptic cross-section along the waveguide. The concept of a \enquote{quantum waveguide} may also be generalised to hypersurfaces that are boundaries of such tubular neighbourhoods. 
In~\cite{HLT} several such examples are analysed in detail.
By our present work these problems are reduced to the calculation of the induced metric as well as the effective operator. Similarly, the treatment of submanifolds of Riemannian manifolds is possible using the techniques of Wittich~\cite{Wi} (see also~\cite[Chapter 3]{Lam}).

An effective operator is derived in~\cite{BMT, deO, deOVe, FrSo, Wi} in the sense of resolvent convergence. In the works~\cite{BMT, deO, Wi} this is convergence of $\eps^{-2}(H-\lambda_0)$ to (the leading order of) $\eps^{-2}(H_\mathrm{a}-\lambda_0)$ for the ground state band $\lambda_0$. The validity of these results hence depends on the fact that the limiting object is independent of $\eps$. This means that the typical energy scale of $H_\mathrm{a}$ must be $\eps^2$, which is generally only the case in $\lambda_0\equiv \text{const}$ (see also the discussion of small energies in Section~\ref{sect:low}). De Oliveira and Verri~\cite{deOVe} treat the situation where $\lambda_0$ has a unique, non-degenerate minimum and this scaling is of order $\eps$. We see an advantage of our approach in the fact that a priori we do not place any restrictions on the behaviour of $\lambda$, and that we can treat also bands different from the ground state. Since our statements are all asymptotic in nature we can naturally 
establish 
approximations beyond the leading order determined by the resolvent limit. So far such refinements were given only by Duclos and Exner~\cite{DEbound}, for a special case where $[\Delta_h, P_0]=0$ and the error is purely due to $\eps H_1$, and in~\cite{WaTeConst} for simple eigenbands and with errors of order $\eps^3$.

\section{Adiabatic theory on fibre bundles}
\subsection{Riemannian submersions of bounded geometry}\label{sect:geom}
In this section we spell out the conditions we pose on the geometry and establish their key consequences. All of our manifolds will satisfy some form of bounded geometry, adapted to their respective structures. The following definition of bounded geometry for manifolds with boundary (or $\partial$-bounded geometry) was introduced by Schick~\cite{Schi}.
\begin{definition}\label{def:delta-bg}
 A Riemannian manifold $(M,g)$ with boundary $\partial M$ is a \textit{$\partial$-manifold of bounded geometry} if the following hold:
 \begin{itemize}
\item \textit{Normal collar}: Let $\nu$ be the inward pointing unit normal of $\partial M$. There exists $r_c>0$ such that the map 
\begin{equation*}
b{:}\,\partial M \times [0,r_c) \to M\,, \qquad (p,t)\mapsto \exp_p(t \nu)
\end{equation*}
is a diffeomorphism to its image.
\item \textit{Injectivity radius of the boundary}: The injectivity radius of $\partial M$ with the induced metric is positive, $r_i(\partial M, g\vert_{\partial M})>0$.
\item \textit{Injectivity radius in the interior}: There is $r_i>0$ such that for $p\in M$ with $\dist(p,\partial M) > r_c/3$ the exponential map restricted to $B(r_i, 0)\subset T_p M$ is a diffeomorphism onto its range.
\item \textit{Curvature bounds}: The curvature tensor of $M$ and the  second fundamental form $S$ of $\partial M$ are $\mathscr{C}^\infty$-bounded tensors on $M$ and $\partial M$ respectively. That is, for every $k\in\N$ their covariant derivatives of order $k$ have $g$-norms bounded by a constant $C(k)$.
\end{itemize}
\end{definition}
If the boundary of $M$ is empty, so are all the conditions on it and the definition reduces to the usual one as given by Eichhorn~\cite{EichBG}. A (vector-) bundle of bounded geometry is usually defined by requiring bounds on trivialisations or transition functions. We adapt this idea here to define more general fibre bundles of bounded geometry.
\begin{definition}
\label{def:unitriv}
Let $(B,g_B)$ be a geodesically complete manifold of bounded geometry. A Riemannian submersion $F\to (M,g) \stackrel{\pi}{\to}(B,g_B)$ is \textit{uniformly locally trivial} if there exists a metric $g_0$ on $F$ such that for every $x\in B$ and metric ball $B(r,x)$ of radius $r<r_i(B)$ there is a trivialisation
\begin{equation*}
 \Phi{:}\,\big(\pi^{-1}(B(r,x)),g\big)\to \big(B(r,x)\times F,g_B\times g_0\big)\,,
\end{equation*}
with $\Phi_{*}$ and $\Phi^*$ bounded with all their covariant derivatives, uniformly in $x$ and $r$.
\end{definition}
We say that such a fibre bundle is of bounded geometry if $(F,g_0)$ is of $\partial$-bounded geometry, which will always be the case since we consider compact fibres. Of course a fibre bundle is also a manifold and the two conflicting notions of bounded geometry are reconciled by:
\begin{lem}
 Let $(M,g)\stackrel{\pi}{\to}(B,g_B)$ be uniformly locally trivial fibre bundle whose fibre $(F,g_0)$ is of bounded geometry. Then the total space is a manifold of bounded geometry in the sense of Definition~\ref{def:delta-bg}.
\end{lem}
The proof of this statement can be found in~\cite{Lam}. The curvature bounds are a straightforward consequence of the definition, while the bounds on injectivity radii can be proved by a technique reminiscent of Cheeger's lemma~\cite{Chee}.

When dealing with the rescaled family $(M,g_\eps)$ it is crucial to note that the estimates on geometric quantities required in Definition~\ref{def:delta-bg} can only become better as $\eps$ decreases. The curvatures for example are completely fixed when they only concern the vertical directions, while they converge to zero (in $g$-norm) if they are associated with at least one horizontal direction. We thus have:
\begin{lem}\label{lem:unifbg}
 Let $M\stackrel{\pi}{\to} B$ be uniformly locally trivial and $F$ compact. Then $(M,g_\eps)=(M,g_F + \eps^{-2}\pi^*g_B)$ satisfies Definition~\ref{def:delta-bg} with the same constants $\lbrace r_c, r_i(\partial M), r_i, C(k): k\in \mathbb{N} \rbrace$ as $(M,g)$.
\end{lem}
From now on we will always assume the following:
\begin{cond}\label{cond:geom}
\leavevmode
\begin{itemize}
\item $F$ is compact,
\item $(B,g_B)$ is of bounded geometry and geodesically complete,
\item $M\stackrel{\pi}{\to}B$ is uniformly locally trivial.
\end{itemize}
\end{cond}
We will call tensors $\mathscr{C}^\infty$-bounded on $M$ if all their $g$-covariant derivatives have bounded $g$-norm. Because of the uniformity of the trivialisations this is equivalent to having bounded derivatives locally on $U\times F$. 

Since $(B,g_B)$ is of bounded geometry and complete, we can choose ${r<r_i(B)}$ 
and an open cover $\mathfrak{U}=\lbrace U_\nu: \nu \in \N \rbrace$ of $B$ by 
geodesic balls of radius $r$ in such a way that any one of these balls 
intersects at most $N(\mathfrak{U})$ others. For every such ball we have 
geodesic coordinates, an orthonormal frame of $\mathscr{C}^\infty$-bounded 
vector fields $\lbrace X_i^\nu: 1\leq i\leq d \rbrace$ on $U_\nu$, obtained by 
radial parallel transport of the basis defining the coordinates, and a 
trivialisation $\Phi_\nu$ of $\pi^{-1}(U_\nu)$, bounded as required by 
Definition~\ref{def:unitriv}. We may also choose a partition of unity $\lbrace 
\chi_\nu: \nu \in \N \rbrace$ subordinate to $\mathfrak{U}$ in such a way that 
all of these objects are $\mathscr{C}^\infty$-bounded uniformly in $\nu$. We fix 
this data related to the cover and will later refer to it simply as 
$\mathfrak{U}$\label{not:U} (for the details of these constructions 
see~\cite{EichBG}).

We can use this cover to define ($L^2$-) Sobolev spaces on $M$ adapted to the bundle structure and the scaled metric $g_\eps$. For this, first define Sobolev spaces on $(F,g_0)$ using a fixed finite cover and, for $x\in U_\nu$, let $\rho^2_{\nu,x}$ be the density $({\Phi_\nu}_*\mathrm{vol}_{F_x})/\mathrm{vol}_{g_0}$ on $F$.
\begin{definition}\label{def:W}
For $\psi\in\mathscr{C}^\infty (F_x)$ and $k\in \mathbb{N}$ put
\begin{equation*}
\norm{\psi}_{W^k_\nu(F_x)}:=\norm{(\Phi_{\nu*}\psi)\rho_{\nu,x}}_{W^{k}(F,g_0)}\,.
\end{equation*}
Denote multiindices by $\alpha \in \mathbb{N}^d$ and define the norm
\begin{equation}\label{eq:normdef}
\norm{\psi}_{W^{k}_\eps(M)}^2:=\sum_\nu \sum_{\abs{\alpha}\leq k} \int_{U_\nu} 
\Big\lVert\eps^{\abs{\alpha}}\prod_{i\leq d}(\Phi_\nu^* X_i^\nu)^{\alpha_i}\chi_\nu \psi
\Big\rVert^2_{W^{k-\abs{\alpha}}_\nu(F_x)}
\mathrm{vol}_{g_B}(\ud x)
\end{equation}
and the \textit{Sobolev space $W^k_{\eps}$} as the completion of $\mathscr{C}_0^\infty (M)$ under this norm. Define $W^{k}_{0,\eps}(M)$ as the closure of $\mathscr{C}_0^\infty (M\setminus\partial M)$ in $W_\eps^k(M)$.
\end{definition}
Sobolev spaces on manifolds of $\partial$-bounded geometry were introduced by 
Schick \cite{Schi2} using normal coordinates of the metric. The virtue of our 
definition is that the same coordinate maps are used for every $\eps$ and that 
the different scaling of vertical and horizontal directions is incorporated in 
an explicit way. However since $(M, g_\eps)$ is of bounded geometry uniformly in 
$\eps$ these definitions are essentially equivalent, apart from a factor 
$\eps^d$ relating the volume measures of $g_B$ and $\eps^{-2}g_B$. That is, 
there is a constant $C(k,\mathfrak{U})>0$ such that
\begin{equation*}
 C^{-1}\norm{\psi}_{W^k_\eps(M)}\leq \eps^{d} \norm{\psi}_{W^k(M,g_\eps)} \leq C\norm{\psi}_{W^k_\eps(M)}\,.
 \end{equation*}
In particular we have $W^0_\eps(M)=L^2(M,g)$, with $\eps$-independent and equivalent norms.

An important consequence of Lemma~\ref{lem:unifbg} and the work of Schick is that the Laplacians $\Delta_{g_\eps}$, with Dirichlet conditions on the boundary, satisfy elliptic inequalities on the spaces $W^k_\eps$ in a uniform way.
\begin{thm}[\cite{Schi2}]\label{thm:ellipt}
For every $k\in \N$ there is a constant $C>0$ such that for every $\psi\in W^2_\eps(M) \cap W^1_{0,\eps}(M)$ with $\Delta_{g_\eps}\psi \in W^k_\eps(M)$ we have $\psi \in W^{k+2}_\eps(M)$ and
 \begin{equation*}
\norm{\psi}_{W^{k+2}_\eps}^2 \leq C ( \norm{\Delta_{g_\eps} \psi}^2_{W^k_\eps} + \norm{\psi}_\mathscr{H}^2)\,.  
 \end{equation*}
\end{thm}

\subsection{Adiabatic and super-adiabatic projections}\label{sect:proj}
In this section we will give an explicit construction of the projections $P_\eps$ of Theorem~\ref{thm:proj}. Again denote by $H$ the operator
\begin{equation*}
 H=-\Delta_{g_\eps} + V + \eps H_1\,,
\end{equation*}
with domain $D(H)=W^2_\eps(M)\cap W^1_{0,\eps}(M)$. More precisely this is an $\eps$-dependent family, since $-\Delta_{g_\eps}$ explicitly depends on $\eps$ and so may $V$ and $H_1$, although we do not make this explicit in the notation.  We assume these satisfy:
\begin{cond}\label{cond:H}
\leavevmode
\begin{itemize}
\item The potential $V\in \mathscr{C}^\infty_b$ is smooth and bounded with all its derivatives, uniformly in $\eps$.
\item $H_1$ is a smooth differential operator of second order and symmetric on 
$D(H)$. It is bounded independently of $\eps$ as a map  from $W^{m+2}_\eps$ to 
$W^{m}_\eps$, for every $m \in \N$.
\end{itemize}
\end{cond}
Under these conditions $H$ is self-adjoint on $D(H)$ and bounded below uniformly in $\eps$ by the Kato-Rellich theorem (see Reed and Simon~\cite[Theorem X.12]{ReSi2}). From now on $H$ will always denote this self-adjoint operator, while expressions like $\Delta_{g_\eps}$ or $H_1$ may also stand for a differential operator without reference to a specific domain. We also denote by $\lambda$ an eigenband of $H_F$ satisfying Condition~\ref{cond:gap} and by $P_0$ the corresponding fibre-wise spectral projection.

A first step in the construction of $P_\eps$ consists in proving that $[H,P_0]=\mathcal{O}(\eps)$ in a suitable sense. Since $P_0$ commutes with $H_F$, and $\eps H_1$ is itself of order $\eps$, this amounts to proving that
\begin{equation*}
 [-\eps^2\Delta_h, P_0]=\mathcal{O}(\eps)\,.
\end{equation*}
Since $P_0$ is fibre-wise and $\Delta_h$ is local we can examine this operator over an open set $U\in \mathfrak{U}$ (defined in Section~\ref{sect:geom}) using a local expression for $\Delta_h$. For a vector field $X\in  \Gamma(TB)$ let $X^*$ denote the unique horizontal vector field, a section of $NF\subset TM$, satisfying $\pi_*X^*=X$. We call this the \textit{horizontal lift} of $X$. Given the orthonormal frame of vector fields $(X_i)_{i\leq d}$ over $U$, $(X_i^*)_{i\leq d}$ is an orthonormal frame of $NF\vert_{\pi^{-1}(U)}$ and we can express $\Delta_h$ as
\begin{equation*}
 \Delta_h\vert_{\pi^{-1}(U)}= \sum_{i=1}^d X_i^* X_i^* - \nabla_{X_i^*} X_i^*  - g_B(\pi_*\eta, X_i)X_i^*\,,
\end{equation*}
where again $\eta$ denotes the mean curvature vector of the fibres.
Thus for our purposes it is sufficient to control commutators of the form $[X^*, P_0]$ for vector fields $X$ on $B$ of bounded length. One might think of calculating such an object by first calculating $[X^*, H_F]$ and then using functional calculus. We must warn here however, that due to the presence of the boundary this commutator is ill-defined. Since if $X^*$ is not tangent to the boundary, its application destroys the Dirichlet condition. For this reason we need to use vector fields that are adapted to the boundary. These are naturally obtained from local trivialisations, taking for $X\in \Gamma(TU)$ the field $\Phi^* X$ on $\pi^{-1}(U)$, which is tangent to the boundary of $M$ (because $\Phi$ also provides a trivialisation $\partial M \cap \pi^{-1}(U) \cong U\times \partial F$). Since this projects to $X$ we have that
\begin{equation*}
 X^*- \Phi^*X=Y \in \ker \pi_*
\end{equation*}
is a vertical vector field. By the boundedness of the geometry of $M$, both $X^*$ and $\Phi^*X$ are $\mathscr{C}^\infty$-bounded if $X$ is, and then so is their difference $Y$.
Now the basic idea is to calculate $[\Phi^*X, P_0]$ using functional calculus and to control $[Y, P_0]$ using the fact that $P_0$ is a spectral projection of $H_F$.
\begin{exm}\label{ex:vfields}
To illustrate the objects we have just discussed we calculate them in a simple example.
Let $h\in \mathscr{C}_b^\infty(\R)$ be a positive function and let $M=\R\times [0, 1 + h]\subset \R^2$, with $B=\R$ and $F= [0,1]$. Let $g_\eps=\eps^{-2}\ud x^2 + \ud y^2$ be the restriction of the rescaled metric on $\R^2$ and $H=-\Delta_{g_\eps}= -\eps^2\partial_x^2 - \partial_y^2$ on $D(H)$. The horizontal lift of $\partial_x \in \Gamma(T\R)$ is trivial $\partial_x^*=\partial_x$, so on $\mathscr{C}^\infty(M)$ we have $[\partial_x, \partial_y^2]=0$. 
A global trivialisation of $M$ is given by 
\begin{equation*}
\Phi{:}\,M\to \R \times [0,1]\,;\qquad (x,y)\mapsto (x,z)=\big(x,y/(1+h(x))\big)\,.
\end{equation*}
For $f\in \mathscr{C}^\infty(M)$ one easily calculates 
\begin{equation*}
\Phi^*\partial_x f= \partial_x f\big(x, (1+ h(x))z\big) = \partial_x f + h'z\partial_y f=\partial_xf + h'y/(1+h)  \partial_y f\,, 
\end{equation*}
so we can identify $Y:=\partial_x^* - \Phi^*\partial_x=-\log(1+h)'y\partial_y$. Clearly $\Phi^*\partial_x$ is tangent to $\partial M$, so for any $f\in \mathscr{C}^\infty(M)$ that vanishes on $\partial M$, $\Phi^*\partial_x f$ is also zero on $\partial M$. On such functions we thus have
\begin{equation*}
 [\Phi^*\partial_x, \partial_y^2]=-[\Phi^*\partial_x, H_F]=[\log(1+h)'y\partial_y, \partial_y^2]= -2\log(1+h)'\partial_y^2\,. 
\end{equation*}
We can observe here that $[\Phi^*\partial_x, H_F]$ is bounded relatively to $H_F$, which will hold in general.
\end{exm}
Since we want to use functional calculus to calculate $[\Phi^*X, P_0]$, control of the resolvent is crucial.
The commutator estimate of the following lemma relies on the fact that $[\Phi^*X_i, Y]$ is a vertical field, so that $[\Phi^*X_i, H_F]$ is bounded by $H_F$ just as in Example~\ref{ex:vfields}.
\begin{lem}\label{lem:res}
Let $U\in \mathfrak{U}$, take $z\in \C$ with $\dist\big(z, \sigma(H_F\vert_U)\big)\geq \delta>0$ and set
\begin{equation*}
R_F(z):=(H_F-z)^{-1}\,.
\end{equation*}
Let $(X_i)_{i\leq d}$ be the orthonormal frame corresponding to $U$ (cf. page~\pageref{not:U}). Then
\begin{equation*}
 [\Phi^* X_i, R_F(z)]\in L^\infty\big(\mathscr{L}(\mathscr{H}_F, D(H_F))\big\vert_U\big)
\end{equation*}
is bounded uniformly in $U$ and $i$. 
\end{lem}
\begin{proof}
We use the trivialisation $\Phi$ to perform the calculations on $U\times F$. For this purpose, endow this set with the metric $\tilde g=\Phi_*g_F+g_B$, induced by $\Phi$ and choosing the canonical lift to the product as the horizontal direction. Then the map $W{:}\,L^2(U\times F, \tilde g)\to L^2 (\pi^{-1}(U),g)$ given by $f\mapsto f\circ \Phi$ is unitary. 
By definition of the vector bundles $D(H_F)$ and $\mathscr{H}_F$ and the bounds on $\Phi$, $W$ also induces isomorphisms
\begin{align*}
 &L^\infty\Big(\mathscr{L}(\mathscr{H}_F, D(H_F))\vert_U\Big) \to L^\infty\left(U,\mathscr{L}(L^2(F),W^2(F)\cap W^1_0(F))\right)\\
 &L^\infty\Big(\mathscr{L}(D(H_F),\mathscr{H}_F)\vert_U\Big) \to L^\infty\left(U,\mathscr{L}(W^2(F)\cap W^1_0(F),L^2(F))\right)\,,
\end{align*}
by conjugation.
We have $W X_i W^*=\Phi^*X_i$ and $W^* \Delta_F W= \Delta_{g_{F}}$,
where the latter is defined as the operator-valued function $x\mapsto \Delta_{g_{F_x}}$ on $U$. Thus
\begin{equation*}
 \left[\Phi^*X_i,R_F(z)\right]=W[X_i,(W^* H_F W - z)^{-1}]W^*\,, 
\end{equation*}
with $W^* H_F W =-\Delta_{g_{F}} + \Phi_*V$. Denote ${R(x,z):=(-\Delta_{g_{F_x}}+\Phi_*V -z)^{-1}}$. The commutator $[X_i, R]$ equals the Lie-derivative $\mathcal{L}_{X_i} R$, so we need to show that $R(x,z)$ depends differentiably on $x\in U$. We have
\begin{align}
\mathcal{L}_{X_i} R(x,z)
&= - R(x,z)\big(\mathcal{L}_{X_i} W^* H_F W \big)R(x,z)\label{eq:commR}\,,
\end{align}
which means that it is enough to show differentiability of $W^* H_F W$.
In order to see that $\mathcal{L}_{X_i} W^* H_F W$ defines a bounded operator $W^2(F)\cap W^1_0(F)\to L^2(F)$, let $x_0 \in U$ and $\phi_{X_i}$ be the flow of $X_i$ on $U\times F$. Then for $0\leq t < T$, the expression $W^* H_F W \circ \phi_X^{t*}\big\vert_{\lbrace x_0 \rbrace \times F}$ makes sense as a one-parameter family of operators $W^2(F)\cap W^1_0(F)\to L^2(F)$ since the domain is invariant under $\phi_{X_i}$. 

Now let $\gamma(t)$ be the integral curve of $X_i$ starting at $x_0$. Since $F$ is compact we can check differentiability locally, so  take an open set $U_F \subset F$ equipped with an orthonormal frame of vector fields $(Y_j)_{j\leq n}$ with respect to $g_{F_{x_0}}$. We extend these to $\gamma \times U_F$ by parallel transport with respect to the Levi-Cività connection $\tilde \nabla$ of $ \tilde g$ and claim that this gives an orthonormal frame of vertical fields. In fact, orthonormality is clear since parallel transport is an isometry. To check that they remain vertical, we calculate 
their component in the direction of any $X_k$, $k\in \lbrace 1, \dots, d \rbrace$. The equation
\begin{equation*}
X_i\tilde g(Y_j,X_k)=\tilde g(\underbrace{\tilde\nabla_{X_i} Y_j}_{=0}, X_k) + \tilde g(Y_j,\tilde\nabla_{X_i} X_k)=\tilde g(Y_j,\tilde\nabla_{X_i} X_k)
\end{equation*}
means that this component satisfies a first-order differential equation, with the initial value given by zero. But $\tilde\nabla_{X_i} X_k$ is horizontal for the metric $\tilde g$, as one easily checks using the Koszul formula. Hence the unique solution to the equation with the given initial value is zero, which means that the fields $Y_j(t)$ are vertical for every $t$. Thus we have
\begin{equation*}
 \Delta_{g_{F}}\vert_{\gamma \times U_F}= \sum_{j=1}^n Y_j\circ Y_j - \nabla_{Y_j} Y_j\,,
\end{equation*}
where $\nabla$ is the Levi-Cività connection of $(F,g_{F})$. Then the Lie derivative equals
\begin{align}
 \mathcal{L}_{X_i}\Delta_{g_{F}}\vert_{\lbrace x_0 \rbrace \times U_F}
 &=\sum_{j=1}^n [X_i,Y_j]Y_j + Y_j [X_i,Y_j] - [X_i,\nabla_{Y_j} Y_j]\notag\\
&=-\sum_{j=1}^n  \big(\tilde\nabla_{Y_j}X_i\big) Y_j + Y_j \big(\tilde\nabla_{Y_j}X_i\big) + [X_i,\nabla_{Y_j} Y_j]\,.\label{eq:L_X H_F}
\end{align}
Now $[X_i, Y_j]$ is a vertical field and by~\eqref{eq:L_X H_F} its coefficients with respect to the basis $(Y_k)_{k\leq n}$ are given by the second fundamental form of ${F\hookrightarrow \lbrace{x_0}\rbrace \times F}$. 
Hence this defines a second order differential operator $W^2(F)\to L^2(F)$, with norm bounded uniformly in $i$ and $U$ by the global bounds on $\Phi$ and $X_i$. The derivative of $V$ is of course just given by $X_i \Phi_*V$, which is bounded for the same reasons and $V\in \mathscr{C}^\infty_b(M)$. 
Finally, by the standard estimate
\begin{equation*}
\norm{R(x,z)}^2_{\mathscr{L}(L^2,W^2)} \leq  2+ (1+2\abs{z}^2)\delta^{-2}\,,
\end{equation*}
the composition~\eqref{eq:commR} defines a uniformly bounded operator $L^2(F)\to W^2(F)$, with image in $W^2(F)\cap W^1_0(F)$. The bounds on $\Phi$ also assure that this still holds after applying the unitary $W$.
\end{proof}
\begin{lem}\label{lem:P_0}
 $\mathcal{E}:=P_0 \mathscr{H}_F$ is a finite rank subbundle of $\mathscr{H}_F$. Moreover, for any $U\in \mathfrak{U}$ and corresponding vector field $X_i$, $i\in\lbrace{1,\dots,  d}\rbrace$
 \begin{equation*}
 [\Phi^* X_i, P_0 ]\in L^\infty\big(\mathscr{L}(\mathscr{H}_F, D(H_F))\big\vert_U\big)
\end{equation*}
is bounded uniformly in $U$ and $i$. In particular $\lambda \in \mathscr{C}^1_b(B)$. 
\end{lem}
\begin{proof}
 Let $\delta>0$ be the gap constant of Condition~\ref{cond:gap}. Let $x_0 \in U\in \mathfrak{U}$ and $\gamma$ be the circle of radius $\delta$ around $\lambda(x_0)$ in $\C$. Now there is an open neighbourhood $U_\delta \subset U$ of $x_0$, such that $\dist\big(\gamma, \sigma(H_F(x))\big)>\delta/2$ for every $x\in U_\delta$, and $P_0$ is given by the Riesz formula
 \begin{equation}\label{eq:Riesz}
  P_0=\frac{\ui}{2\pi} \int_\gamma R_F(z)\ud z\,.
 \end{equation}
 Mapping this to $U_\delta\times F$ with the unitary $W$ from the proof of Lemma~\ref{lem:res}, we immediately see that $P_0$ is strongly continuous in $x$ because this holds for $W H_F(x) W^*$ and $R(x,z)$.
This implies continuity of the projected transition maps of the bundle $\mathscr{H}_F$ (cf.~\cite[appendix B]{Lam}), so $P_0 \mathscr{H}_F$ is a subbundle.

The statement on the commutator $[\Phi^* X_i, P_0 ]$ is a direct consequence of the Riesz formula and Lemma~\ref{lem:res}. To check that this implies $\lambda \in \mathscr{C}^1_b(B)$, let $k=\mathrm{rank}(\mathcal{E})$ and observe that $\lambda=k^{-1}\tr(H_FP_0)$, where the trace is taken in the fibre of $\mathscr{H}_F$. This is continuous for the same reasons as $P_0$. Now we may calculate $(X_i\lambda) (x_0)$ by lifting to $\pi^{-1}(U)$:
\begin{align}
 \pi^*(X_i\lambda)&=  [\Phi^*X_i, \pi^*\lambda]\notag\\
 &= k^{-1}\, \tr{([\Phi^* X_i,H_F P_0])}\notag\\
 &=k^{-1} \, \tr{([\Phi^*X_i,H_F]P_0 + H_F P_0 [\Phi^*X_i,P_0] + H_F [\Phi^*X_i,P_0]P_0 )}\label{eq:difflambda}\,.
 \end{align}
 All of these terms are trace-class since they have finite rank. They are also uniformly bounded since
 \begin{equation*}
  [\Phi^*X_i,H_F]\in L^\infty\big(\mathscr{L}(D(H_F),\mathscr{H}_F )\big\vert_U\big)
 \end{equation*}
is uniformly bounded by~\eqref{eq:L_X H_F}. The terms are continuous in $x$ by the same reasoning as for $P_0$, so since $x_0$ and $i$ were arbitrary this implies $\lambda \in \mathscr{C}^1_b(B)$. 
\end{proof}
In order to control $[\Delta_h,P_0]$ we also need to take care of commutators of $P_0$ with two horizontal vector fields. Then, in our iterative construction of $P_\eps$, commutators with an arbitrary number of such fields may appear. Additionally we will need to keep track of boundary values in order to be sure when we have an object compatible with the domain of $H$. For a systematic discussion of these issues we define special algebras of differential operators. These differential operators will have coefficients in $L^\infty(\mathscr{L}(\mathscr{H}_F))$, which are exactly the fibre-wise operators in $\mathscr{L}(\mathscr{H})$. We assume these coefficients to be smooth in the following sense:

Take $U_\nu \in \mathfrak{U}$ and let $\mathcal{C^\nu}\subset L^\infty\left(\mathscr{L}(\mathscr{H}_F)\vert_{U_\nu}\right)$\label{not:coeff} be those linear operators $A$ for which any commutator of the form
\begin{equation}\label{eq:commC}
\left[\Phi^*_\nu X^\nu_{i_1},\dots,[\Phi^*_\nu X^\nu_{i_k},A]\cdots\right]
\end{equation}
defines an element of $L^\infty\left(\mathscr{L}(\mathscr{H}_F)\vert_{U_\nu}\right)$, 
where $k \in \N$ and $i_1,\dots, i_k \in \lbrace 1,\dots,d\rbrace$.

Let $\mathcal{C}_H^\nu\subset \mathcal{C}^\nu$ be the subset 
for which~\eqref{eq:commC} belongs to $L^\infty\left(\mathscr{L}(\mathscr{H}_F,D(H_F))\vert_{U_\nu}\right)$. This is equivalent to saying that $A\in\mathcal{C}_H^\nu$ if and only if $H_F A\in \mathcal{C}^\nu$, as can be seen from the proof of Lemma~\ref{lem:res}.
\begin{definition}\label{def:algebra}
The algebras $\mathcal{A}$ and $\mathcal{A}_H$ consist of those linear operators in $\mathscr{L}(W^\infty(M), \mathscr{H})$ satisfying
\begin{equation*}
\forall f\in W^\infty(M):\pi(\supp Af)\subset \pi(\supp f)\,
\end{equation*}
and
\begin{equation*}
 A\vert_{\pi^{-1}(U_\nu)}= \sum_{\alpha\in\mathbb{N}^d} A_\alpha^\nu(\eps) \eps^{\abs{\alpha}} (\Phi^*X^\nu_1)^{\alpha_1}\cdots(\Phi^*X^\nu_d)^{\alpha_d}\,,
\end{equation*}
with $A_\alpha^\nu\in \mathcal{C}^\nu$, respectively $\mathcal{C}_H^\nu$, for which there exists $\ell\in \N$ such that
$A_\alpha^\nu=0$ for all $\abs{\alpha}>\ell$, $\nu\in \N$ and furthermore there exist constants $C(\alpha, k)$ such that
\begin{equation*}
 \left\Vert\left[\Phi^*_\nu X^\nu_{i_1},\dots,[\Phi^*_\nu X^\nu_{i_k},A^\nu_\alpha(\eps)]\cdots\right]\right\Vert_{\mathscr{L}(\mathscr{H}_F)} \leq C(\alpha, k)
\end{equation*}
for all $\nu, k\in \N$, $i_1,\dots, i_k \in \lbrace 1,\dots,d\rbrace$ and $\eps>0$.
\end{definition}
From now on we write $\mathcal{C}_\bullet^\nu$ and $\mathcal{A}_\bullet$ in statements that hold with or without the subscript $H$. $\mathcal{A}_\bullet$ is an algebra because of the commutator condition~\eqref{eq:commC} for $\mathcal{C}_\bullet^\nu$ and $[\Phi^*X_i, \Phi^*X_j]=\Phi^*[X_i, X_j]$, allowing us to arrange the vector fields in any order without producing vertical derivatives.
 $\mathcal{A}_H$ consists of those $A\in \mathcal{A}$  whose image consists of functions satisfying the Dirichlet condition and for which $H_F A \in \mathcal{A}$. Hence $\mathcal{A}_H\mathcal{A}\subset \mathcal{A}_H$ and $\mathcal{A}_H$ is a right ideal of $\mathcal{A}$.

 The algebra $\mathcal{A}_\bullet$ is filtered by setting 
\begin{equation*}
\mathcal{A}^k_\bullet:=\left\lbrace A\in \mathcal{A}_\bullet:\forall \nu\in\N\, \big(\abs{\alpha}>k \Rightarrow A_\alpha^\nu=0\big)\right\rbrace\,.
\end{equation*}
 Clearly $\mathcal{A}^k\subset \mathscr{L}(W^k_\eps,\mathscr{H})$ so it inherits this operator norm, which we denote by $\norm{\cdot}_k$. It is because of the $\eps$ dependence of this norm~\eqref{eq:normdef} that we explicitly introduced the factors of $\eps^{\abs{\alpha}}$ into the definition of $\mathcal{A}_\bullet$. An additional filtration is given by the order in $\eps$ by defining $\mathcal{A}_\bullet^{j,\ell}$ to be those $A\in \mathcal{A}^j_\bullet$ for which the constants $C(\alpha,k)$ of Definition~\ref{def:algebra} can be chosen of order $\eps^\ell$:
\begin{equation*}
 \mathcal{A}^{j,l}_\bullet:=\left\lbrace A\in \mathcal{A}^j_\bullet: \eps^{-\ell}A\in \mathcal{A}^j_\bullet \right\rbrace
\end{equation*}

 This of course implies that $\norm{A}_{j}=\mathcal{O}(\eps^\ell)$. 
 Note that a differential operator of order $k$ is also one of order $k+1$, so $\mathcal{A}^k_\bullet\subset \mathcal{A}^{k+1}_\bullet$, while a norm of order $\ell+1$ is also of order $\ell$, so $\mathcal{A}^{k,\ell+1}_\bullet\subset \mathcal{A}^{k,\ell}_\bullet$. We may also observe that due to the commutation properties of the coefficients and vector fields we have the composition property $\mathcal{A}^{j,k}_\bullet \mathcal{A}^{\ell,m}\subset \mathcal{A}^{j+\ell,k+m}_\bullet$. More precisely we have for $A\in \mathcal{A}^{j}$, $B \in \mathcal{A}^{\ell}$
 \begin{equation*}
 AB\vert_{\pi^{-1}(U_\nu)}-\sum_{\substack{\abs{\alpha}=k \\ \abs{\beta}=l}} A_\alpha^\nu B_\beta^\nu \eps^{k+\ell}
  (\Phi^*X^\nu_1)^{\alpha_1 + \beta_1}\cdots(\Phi^*X^\nu_d)^{\alpha_d + \beta_d}
  \in \mathcal{A}^{j+\ell -1}\vert_{\pi^{-1}(U_\nu)}\,,
 \end{equation*}
 and we may also note that terms containing commutators of $\Phi^*X_i$ with other vector fields or the coefficients $A_\alpha, B_\beta$ produce terms of lower order in $\eps$.   
 \begin{rem}\label{rem:normA^k}
 The condition $\pi (\supp Af) \subset \pi(\supp f)$ allows us to calculate the norms $\norm{\cdot}_k$ locally with respect to the base since (denoting by $N(\mathfrak{U})$ the multiplicity of $\mathfrak{U}$, see page~\pageref{not:U})
\begin{align*}
\lVert A\psi\rVert_{W^0_\eps(M)}^2&=\sum_\nu \Vert\chi_\nu A\psi\Vert^2_{\mathscr{H}}
=\sum_\nu \Vert{\chi_\nu A\sum_\mu \chi_\mu\psi}\Vert^2\\
 &\leq N(\mathfrak{U}) \sum_{\mu,\nu} \norm{\chi_\nu A \chi_\mu \psi}^2\\
 &\leq N(\mathfrak{U})^2 \sum _\mu \sup_{\nu} \norm{\chi_\nu A}_{\mathscr{L}(W^k_\eps(\pi^{-1}U_\mu), \mathscr{H})}^2 \norm{\chi_\mu \psi}^2_{W^k_\eps(\pi^{-1}U_\mu)}\\
 &\leq N(\mathfrak{U})^2  \sup_{\mu} \norm{A}_{\mathscr{L}(W^k_\eps(\pi^{-1}U_\mu), \mathscr{H})}^2
 \norm{\psi}^2_{W^k_\eps(M)}\,.
\end{align*}
Thus for any $A\in \mathcal{A}^k$ 
\begin{equation*}
\norm{A}_k \leq N(\mathfrak{U})^{3/2}  \sup_{\mu} \norm{A}_{\mathscr{L}(W^k_\eps(\pi^{-1}U_\mu), \mathscr{H})}\,,
\end{equation*}
where $W^k_\eps(\pi^{-1}U_\nu)$ is defined in the trivialisation by equation~\eqref{eq:normdef} without the sum over $\nu$.
\end{rem}
The starting point for our construction is to show that $R_F(z)$ and $P_0$ are elements of these algebras, following Lemma~\ref{lem:res} and~\ref{lem:P_0}.
\begin{prop}\label{prop:P alg}
Let $z\in \mathscr{C}^\infty_b(B, \C)$ with $\dist\big(z(x), \sigma(H_F(x))\big)\geq \delta >0$, then $R_F(z) \in \mathcal{A}^{0,0}_H$. Furthermore, $P_0 \in \mathcal{A}^{0,0}_H$.
\end{prop}
\begin{proof}
The first statement is shown by iterating the proof of Lemma~\ref{lem:res}, which can be done by the explicit form of the commutator~\eqref{eq:L_X H_F}. The second statement then follows from the gap condition and the Riesz formula~\eqref{eq:Riesz}.
\end{proof}
As we see from the proof of Lemma~\ref{lem:P_0} this immediately gives us a simple corollary.
\begin{cor}\label{cor:lambda}
The eigenband $\lambda$ is smooth and bounded with all its derivatives.
\end{cor}
From $P_0$ and $R_F$ we will be able to construct many other elements of $\mathcal{A}$. The first is the \textit{reduced resolvent}.
\begin{cor}\label{cor:R_F}
 $R_F(\lambda):=(H_F-\lambda)^{-1}(1-P_0)\in \mathcal{A}_H^{0,0}$.
\end{cor}
\begin{proof}
Follows directly from the assertions~\ref{lem:res},~\ref{lem:P_0} and~\ref{cor:lambda} together with the local formula (in the notation of~\eqref{eq:Riesz})
\begin{equation*}
 R_F(\lambda)=(1-P_0)\frac{\ui}{2\pi}\int_\gamma \frac{1}{\lambda-z} R_F(z)\ud z (1-P_0)\,.
\end{equation*}
\end{proof}
A systematic construction of objects in $\mathcal{A}$ is provided by the following lemma.
\begin{lem}\label{lem:commAlg}
 Let $A,B\in \mathcal{A}_H$ with $AB \in \mathcal{A}_H^{k,\ell}$, then 
\begin{equation*}
[\Delta_{g_\eps},A]B\in \mathcal{A}^{k+1,\ell}
\end{equation*} and 
\begin{equation*}
[\eps^2\Delta_h,A]B\in \mathcal{A}^{k+1,\ell+1}\,.
\end{equation*}
\end{lem}
\begin{proof}
We split $\Delta_{g_\eps}=\Delta_F+\eps^2 \Delta_h$ and first observe that 
\begin{equation*}
[\Delta_F,A]B=\underbrace{\Delta_FA}_{\in \mathcal{A}}B-A\underbrace{\Delta_FB}_{\in \mathcal{A}} \in \mathcal{A}^{k,\ell}\,,
\end{equation*}
 since $\Delta_F A_\alpha \in \mathcal{C}$ if $A_\alpha \in \mathcal{C}_H$. Hence the second claim implies the first one.
 
 Since the definition of $\mathcal{A}^k$ and its norm are local with respect to the base (cf. Remark~\ref{rem:normA^k}) it is sufficient to show the claim on $\pi^{-1}(U_\nu)$. We fix $\nu$ and split $X_i^*=\Phi^*X_i+Y_i$. In this frame we have $\eps^2\Delta_h=\eps^2\sum_{i\leq d} \Phi^*X_i\Phi^*X_i+ \eps^2D$, where $D$ contains first order differential operators and second order parts that contain at least 
one vertical derivative.
We have for every $j\in \lbrace 1,\dots, d\rbrace$
\begin{align*}
 \left[\Phi^* X_j, A\right]\big\vert_{\pi^{-1}(U)}= 
 \sum_{\alpha\in\mathbb{N}^d} \eps^{\abs{\alpha}}\Big(&\underbrace{[\Phi^*X_j, A_\alpha]}_{\in \mathcal{C}_H} (\Phi^*X_1)^{\alpha_1}\cdots(\Phi^*X_d)^{\alpha_d}\\
&+ A_\alpha [\Phi^*X_j,(\Phi^*X_1)^{\alpha_1}\cdots(\Phi^*X_d)^{\alpha_d}]\Big) \,.
\end{align*}
This is of the same order as $A$ in $\mathcal{A}_H$ because $[\Phi^* X_j, \Phi^*X_i]=\Phi^*[X_j,X_i]$.
Hence
\begin{equation*}
\chi \sum_{i\leq d} [\eps^2\Phi^*X_i\Phi^*X_i, A]B\in \mathcal{A}^{k+1,\ell+1}_H\,.
\end{equation*}
Now for a $\mathscr{C}^\infty$-bounded vertical field $Y$, the commutator $[\Phi^*X_i, Y]$ is also vertical and bounded, so we have $YAB$ and $AYB\in \mathcal{A}^{k,\ell}$. By commuting all the
 $\Phi^*X_i$ to the right we see that
\begin{equation*}
 \chi [\eps^2 D, A]B\in \mathcal{A}^{k+1,\ell+1}\,.
\end{equation*}
This proves the second claim and thus completes the proof.
\end{proof}
\begin{lem}\label{lem:H_1}
 For every $A\in \mathcal{A}^{k,\ell}_H$ we have $H_1A\in \mathcal{A}^{k+2,\ell}$.
\end{lem}
\begin{proof}
Locally we have
 \begin{align*}
H_1\vert_{\pi^{-1}(U_\nu)}= &\sum_{i\leq j \leq d} A_{ij}^\nu \eps^2\Phi^*X^\nu_i\Phi^*X^\nu_j
+ \sum_{i\leq d} B_i^\nu \eps \Phi^*X^\nu_i + C\,,
\end{align*}
with vertical differential operators $C, B_i^\nu$, of second, respectively, first order. The Lemma follows easily from this in the same way as Lemma~\ref{lem:commAlg}, using the bounds on $H_1{:}\, W^{m+2}_\eps \to W^m_\eps$ to obtain the required uniformity in $\nu$.
\end{proof}
%
%
%
As a consequence of Proposition~\ref{prop:P alg} and Lemmas~\ref{lem:commAlg},~\ref{lem:H_1} we thus have
\begin{equation*}
 [H,P_0]P_0 = 
 \underbrace{[-\eps^2 \Delta_h, P_0]P_0}_{\in \mathcal{A}^{1,1}}  
 +\eps [H_1, P_0]P_0 \in \mathcal{A}^{2,1}\,.
\end{equation*}
Since $P_0$ is a projection, it has the property that
\begin{equation*}
[A, P_0]=[A,P_0^2]= P_0[A,P_0] + [A,P_0]P_0\,, 
\end{equation*}
and thus
\begin{equation*}
P_0[A,P_0]P_0=2 P_0[A,P_0]P_0=0\,.
\end{equation*}
Hence the commutator is off-diagonal with respect to the splitting of $\mathscr{H}=P_0\mathscr{H}\oplus {(1-P_0)\mathscr{H}}$ induced by $P_0$. We will use this property of projections very frequently in the following construction of the super-adiabatic projections. These will be obtained from a sequence $P^N$ of almost-projections in $\mathcal{A}_H$ that have the same asymptotic expansion as $P_\eps$. We construct this sequence explicitly, similarly to~\cite[lemma 3.8]{TeAdiab} but replacing the symbol classes of pseudo-differential calculus by the algebras $\mathcal{A}_\bullet$.
\begin{lem}\label{lem:sadiabatic}
For every $k\in \N$ there exists $P_k \in \mathcal{A}_H^{2^k,0}$, such that
\begin{equation*}
 P^N=\sum_{k=0}^{N} \eps^k P_k
\end{equation*}
 satisfies
\begin{teile}
\item $(P^N)^2-P^N \in \mathcal{A}_H^{2^{N+1},N+1}$,
\item $\norm{\big[H, P^N\big]}_{2^N+2} = \mathcal{O}(\eps^{N+1})$ on $D(H)$.
\end{teile}
\end{lem}
\begin{proof}
Take $P_0$ to be the projection on the eigenband $\lambda$ as above. By Proposition~\ref{prop:P alg} we have $P_0\in \mathcal{A}^{0,0}_H\subset \mathcal{A}^{1,0}_H$ and \emph{1)} is trivially satisfied because it is a projection. For \emph{2)} first observe that by Condition~\ref{cond:H} we have $[H_1,P_0]=H_1P_0-P_0H_1=\mathcal{O}(1)$. To see that 
\begin{equation}\label{eq:commP_0}
\norm{[-\eps^2\Delta_h,P_0]}_2=\mathcal{O}(\eps)\,,
\end{equation}
one just commutes all derivatives of the form $\Phi^*X_i$ to the right as in the proof of Lemma~\ref{lem:commAlg}, so \emph{2)} holds.

We define $P_{N+1}$ recursively by splitting it into diagonal and off-diagonal parts with respect to $P_0$ and prove \emph{1)} and \emph{2)} by induction.
To shorten the notation we write $P_0^\perp:=1-P_0$. Define
\begin{equation*}
\begin{split}
\eps^{N+1}P_{N+1} :=&\underbrace{-P_0 \big((P^N)^2-P^N\big) P_0 + P_0^\perp\big((P^N)^2-P^N\big)P_0^\perp}_{=:\eps^{N+1}P_{N+1}^D}\\
&\quad\underbrace{-P_0^\perp R_F(\lambda)\big[H, P^N\big]P_0 + P_0\big[H, P^N\big]R_F(\lambda) P_0^\perp}_{=:\eps^{N+1}P_{N+1}^O}\,.
\end{split}
\end{equation*}
This is an element of $\mathcal{A}_H^{2^{N+1}}$ because of Lemma~\ref{lem:commAlg} and the fact that $\mathcal{A}_H$ is a right ideal, since $2^{N+1} \geq 2^N + 2$ for $N\geq 1$ and $P_1\in \mathcal{A}_H^{2,0}$ because ${P_0, R_F(\lambda) \in \mathcal{A}_H^{0,0}}$ by~\ref{prop:P alg},~\ref{cor:R_F}.
$P_{N+1}$ is of clearly order $\eps^0$ by application of \emph{1)} and \emph{2)} to $P^N$, which is the induction hypothesis.

\medskip
\textit{Proof of 1)}
We prove this for diagonal and off-diagonal parts separately. In both cases it is just a simple calculation using $P^N={P_0 + \mathcal{A}_H^{2^N,1}}={P_0+\mathcal{O}(\eps)}$.
\begin{itemize}
 \item Diagonal:
\begin{align}
P_0&\big((P^{N+1})^2-P^{N+1}\big) P_0\notag\\
&= P_0\big((P^{N}+\eps^{N+1}P_{N+1})^2-P^{N}-\eps^{N+1}P_{N+1}\big)P_0\notag\\
&=\begin{aligned}[t]&P_0\big((P^N)^2-P^N +\eps^{N+1}\big(P^N P_{N+1} + P_{N+1}P^N - P_{N+1}\big)\big)P_0\\
&+\mathcal{A}_H^{2^{N+2},2N+2}
\end{aligned}\notag\\
&=\overbrace{P_0\big((P^N)^2-P^N\big) P_0 + \eps^{N+1}P_0P_{N+1}^DP_0}^{=0} +\mathcal{A}_H^{2^{N+2},N+2}\notag\\
&\in \mathcal{A}_H^{2^{N+2},N+2}\,.\notag
\end{align}
\item Off-diagonal:
\begin{align}
P_0^\perp&\big((P^{N+1})^2 - P^{N+1}\big) P_0\notag\\
&=\begin{aligned}[t]& P_0^\perp\big((P^N)^2-P^N\big) P_0 +
\eps^{N+1}\overbrace{P_0^\perp\big(P_{N+1}P^N-P_{N+1}\big)P_0}^{\in\mathcal{A}_H^{2^{N+1} + 2^N,1}}\notag\\
&+ 
\eps^{N+1}\underbrace{P_0^\perp P^N P_{N+1} P_0}_{\in\mathcal{A}_H^{2^{N+1} + 2^N,1}}  +\mathcal{A}_H^{2^{N+2},2N+2}
\end{aligned}\\
&= P_0^\perp\big((P^N)^2-P^N\big)\big(P^N + P_0 - P^N \big) P_0 + \mathcal{A}_H^{2^{N+2},N+2}\notag\\
&= P_0^\perp\big((P^N)^2-P^N\big) P^N P_0 + \mathcal{A}_H^{2^{N+2},N+2}\notag\\
&= P_0^\perp P^N\big((P^N)^2-P^N\big) P_0 + \mathcal{A}_H^{2^{N+2},N+2}\notag\\
&\in \mathcal{A}_H^{2^{N+2},N+2}\,.\notag
\end{align}
\end{itemize}
The calculations for the $P_0^\perp$-$P_0^\perp$ and $P_0$-$P_0^\perp$ blocks are basically the same, so \emph{1)} is verified.

\medskip
\textit{Proof of 2)}
\begin{itemize}
 \item Diagonal:
We will only do the calculation for the $P_0$-block. The one for $P_0^\perp$ is similar, one merely needs to commute derivatives to the right as for~\eqref{eq:commP_0}, since Lemma~\ref{lem:commAlg} is not directly applicable.
First we show $P_0[H, \eps^{N+1}P_{N+1}^O]P_0= \mathcal{O}(\eps^{N+2})$:
\begin{align*}
 P_0&\big[H, \eps^{N+1}P_{N+1}^O\big]P_0\notag\\
&=\eps^{N+1}P_0 \big(H P_0^\perp P_{N+1}^O - P_{N+1}^O P_0^\perp H \big)P_0\notag\\
&=\eps^{N+1}
\bigl(-P_0\underbrace{[H,P_0] P_{N+1}^O}_{\in\mathcal{A}^{2^{N+1}+2,1}}P_0  -
P_0P_{N+1}^O P_0^\perp \underbrace{[H,P_0]P_0}_{\in\mathcal{A}^{2,1}}\bigr)\notag\\
&=\mathcal{O}(\eps^{N+2})\,.
\end{align*}
Now by definition $P^{N+1}-\eps^{N+1}P^O_{N+1}=P^N+\eps^{N+1}P_{N+1}^D$, so we still have to calculate
\begin{align*}
P_0&\big[H, P^N + \eps^{N+1}P_{N+1}^D\big]P_0\notag\\
&=P_0\big[H, P^N - P_0((P^N)^2-P^N)P_0\big]P_0\notag\\
&=\begin{aligned}[t]&2 P_0 \big[H, P^N\big]P_0 - P_0\big[H,(P^N)^2\big]P_0\\
&+\underbrace{\big(P_0[H, P_0]\big((P^N)^2-P^N\big)P_0 + P_0\big((P^N)^2-P^N\big)[H, P_0]P_0 \big) }_
{\in \mathcal{A}_H^{2^{N+1}+2,N+2} \text{ by induction hypothesis and~\ref{lem:commAlg},~\ref{lem:H_1}}}
\end{aligned}\\
&=P_0\big(2\big[H, P^N\big]-P^N \big[H, P^N\big] -\big[H, P^N\big] P^N \big) P_0 + \mathcal{O}(\eps^{N+2})\notag\\
&=\underbrace{P_0\big((P_0 - P^N) \big[H, P^N\big]+\big[H, P^N\big](P_0-P^N)\big)P_0}_{\in \mathcal{A}_H^{2^{N+1}+2,N+2}}
+ \mathcal{O}(\eps^{N+2})\\
&=\mathcal{O}(\eps^{N+2})\,.\notag
\end{align*}
\item Off-diagonal:\\
Here we use the statements of~\ref{cor:lambda} and~\ref{lem:commAlg} to get
\begin{equation}
 [-\eps^2\Delta_h+\lambda, P_{N+1}]P_0 \in \mathcal{A}^{2^{N+1}+1,1}\,.\notag
\end{equation}
This gives us
\begin{equation}
[H,P_{N+1}]P_0=[H_F-\lambda,P_{N+1}]P_0+ \underbrace{[-\eps^2\Delta_h+\lambda + \eps H_1, P_{N+1}]P_0}_{\in \mathcal{A}^{2^{N+1}+2,1}}\,.\notag
\end{equation}
We insert this into
\begin{align}
P_0^\perp&\big[H, P^N + \eps^{N+1}P_{N+1}\big]P_0\notag\\
&=P_0^\perp\big(\big[H, P^N\big] + \eps^{N+1}\big[H_F-\lambda,P_{N+1}\big]
\big)P_0+ \mathcal{O}(\eps^{N+2})\notag\\
&=P_0^\perp\big(\big[H, P^N\big] + \eps^{N+1}[H_F-\lambda, P_0^\perp P_{N+1} P_0]\big) P_0 + \mathcal{O}(\eps^{N+2})\notag\\
&=P_0^\perp(\big[H, P^N\big] - \underbrace{(H_F-\lambda)R_F(\lambda)}_{=1}\big[H,P^N\big])P_0 + \mathcal{O}(\eps^{N+2})\notag\\
&= \mathcal{O}(\eps^{N+2})\,,\notag
\end{align}
which completes the proof for the $P_0^\perp$-$P_0$-block. The argument for the other off-diagonal block is the same.
\end{itemize}
\end{proof}
\subsubsection{Proof of Theorem~\ref{thm:proj}}
The proof will use auxiliary energy cut-offs. We require these to satisfy:
\begin{definition}\label{def:cutoff}
 A function  $f\in \mathscr{C}^\infty_0(\R, [0,1])$ is a \textit{regular cut-off} if for every $s\in (0, \infty)$, the power $f^s\in \mathscr{C}^\infty_0(\R, [0,1])$. 
\end{definition}
In particular this prevents these functions from having zeros of finite order. The following lemma on the functional calculus for such functions can be derived from the Helffer-Sjöstrand formula (see~\cite[appendix C]{Lam} for a proof).
\begin{lem}\label{lem:chi}
Let $H$ be self-adjoint on $D(H)\subset \mathscr{H}$. Let $T\in \mathscr{L}(\mathscr{H})\cap\mathscr{L}(D(H))$ be self-adjoint on $\mathscr{H}$.
If $\chi$ is a regular cut-off and
\begin{align*}                                                                                                        
&\norm{[T,H]}_{\mathscr{L}(D(H),\mathscr{H})}=\mathcal{O}(\eps)\\
&\norm{[T,H]\chi^s(H)}_{\mathscr{L}(\mathscr{H})}=\mathcal{O}(\eps^k)\,,
\end{align*}
for some $k\in \mathbb{N}$ and all $s\in (0, \infty)$, then
\begin{teile}
\item $\norm{[T,\chi(H)]}_{\mathscr{L}(\mathscr{H},D(H))}=\mathcal{O}(\eps^k)$;
\item If additionally $T$ is a projection 
\begin{equation*}
\norm{T\chi(THT)-T\chi(H)T}_{\mathscr{L}(\mathscr{H},D(H))}=\mathcal{O}(\eps^k)\,.
\end{equation*}
\end{teile}
\end{lem}
We restate the main point of the theorem for convenience.
\begin{thm*}[\ref{thm:proj}]
 For every $\Lambda>0$ and $N\in \N$ there exists an orthogonal projection $P_\eps \in \mathscr{L}(\mathscr{H})\cap \mathscr{L}(D(H))$ that satisfies
\begin{equation*}
 \norm{\left[H, P_\eps\right] \varrho(H)}_{\mathscr{L}\left(\mathscr{H}\right)} = \mathcal{O}(\eps^{N+1})
\end{equation*}
for every Borel function $\varrho{:}\,\mathbb{R}\to [0,1]$ with support in $(-\infty, \Lambda]$.
\end{thm*}
\begin{proof}
To prove the statement for $N\in \mathbb{N}$ and $\Lambda>0$, take $P^{N}$ from Lemma~\ref{lem:sadiabatic} and let $\chi_1\in \mathscr{C}^\infty_0\left(\mathbb{R},[0,1]\right)$ 
be a regular cut-off, equal to one if $x\in [\inf \sigma (H)-1, \Lambda+1]$ and equal to zero if $x \notin (\inf \sigma(H) -2, \Lambda +2)$. Put $\tilde P:= P^{N}-P_0\in\mathcal{A}^{2^N,1}_H$ and define
\begin{equation*}
P^\chi:= P_0 + \tilde P\chi_1(H) + \chi_1(H)\tilde P\left(1-\chi_1(H)\right)\,.
\end{equation*}
The first step is to justify that $P^\chi=P_0 + \mathcal{O}(\eps)$ in $\mathscr{L}(\mathscr{H})$ and $\mathscr{L}\left(D(H)\right)$.
We have $\chi_1\in \mathscr{L}\big(\mathscr{H}, D(H^{2^{N-1}})\big)$ and by 
elliptic regularity ${D(H^{2^{N-1}})\subset W^{2^N}_\eps}$ 
(cf.~Theorem~\ref{thm:ellipt}), so $\tilde P\chi_1\in 
\mathscr{L}(\mathscr{H})\cap\mathscr{L}\left(D(H)\right)$. Therefore its adjoint 
is also a bounded operator and from the construction of $P_N$ we can see that 
$\chi_1 \tilde P=(\tilde P\chi_1)^*$ on $W^{2^N}_\eps$, so they are equal in 
$\mathscr{L}(\mathscr{H})$ because $W^{2^N}_\eps$ is a dense subspace of 
$\mathscr{H}$. 
Hence $P^\chi\in\mathscr{L}(\mathscr{H})$ is self-adjoint by construction.

We want to prove that also $P^\chi\in \mathscr{L}(D(H))$. To show $\chi_1\tilde P \in \mathscr{L}(D(H))$ we need to show $[H,\chi_1 \tilde P]=\chi_1[H,\tilde P] \in \mathscr{L}(D(H),\mathscr{H})$. But actually, by the same argument as before, we have $\chi_1[H,\tilde P]=([\tilde P,H]\chi_1)^*$ on ${W^{2^N+2}_\eps\cap D(H)}$, and thus $\chi_1[H,\tilde P]\in \mathscr{L}(\mathscr{H})$. These norms are of order $\eps$ because $\tilde P\in \mathcal{A}^{2^N,1}_H$. Consequently $P^\chi-P_0=\mathcal{O}(\eps)$ in $\mathscr{L}(\mathscr{H})$ as well as $\mathscr{L}\left(D(H)\right)$.
We conclude that
\begin{equation}\label{P chi comm1}
\norm{\left[H,P^\chi\right]}_{\mathscr{L}\left(D(H),\mathscr{H}\right)}
= \norm{\left[H,P_0\right]}_{\mathscr{L}\left(D(H),\mathscr{H}\right)} + \mathcal{O}(\eps)
= \mathcal{O}(\eps)\,.
\end{equation}
Now let $\chi_2\in \mathscr{C}^\infty_0\left(\mathbb{R},[0,1]\right)$ be another 
regular cut-off, equal to 
zero where $\chi_1\neq 1$ and
equal to one on $[\inf\sigma(H),\Lambda]$. Then we have 
$\chi_1\chi_2=\chi_2$, $(1-\chi_1)\chi_2=0$ and from Lemma~\ref{lem:sadiabatic} 
we get
\begin{equation}
\label{P chi comm2}
\norm{\left[H,P^\chi\right]\chi_2(H)}_{\mathscr{L}\left(\mathscr{H}\right)}
= \lVert[H,P^{N}]\chi_2(H)\rVert_{\mathscr{L}\left(\mathscr{H}\right)}
= \mathcal{O}(\eps^{N+1})\,.
\end{equation}
Since $P^\chi$ is close to the projection $P_0$ we have for $m\in\lbrace0,1\rbrace$:
\begin{equation*}
\lVert(P^\chi)^2-P^\chi\rVert_{\mathscr{L}\left(D(H^m)\right)}= \mathcal{O}(\eps)\,.
\end{equation*}
Thus there is a constant $C>0$ such that the spectrum of $P^\chi$ (as an operator in $\mathscr{L}(\mathscr{H})$ as well as $\mathscr{L}\left(D(H)\right)$) satisfies
\begin{equation*}
\sigma(P^\chi)\subset [-C\eps, C\eps] \cup [1-C\eps, 1+ C\eps]\,.
\end{equation*}
Take $\gamma$ to be the circle of radius $1/2$ around $z=1$. Then for $\eps < (4C)^{-1}$ the integral
\begin{equation*}
 P_\eps:=\frac{\ui}{2\pi}\int_\gamma \left(P^\chi-z\right)^{-1}\ud z
\end{equation*}
defines an element of $\mathscr{L}(\mathscr{H})$ and $\mathscr{L}\left(D(H)\right)$ with norm less than two. It is an orthogonal projection on $\mathscr{H}$ by the functional calculus and satisfies
\begin{align*}
 P_\eps - P_0 &= \frac{\ui}{2\pi}\int_\gamma \big((P^\chi-z)^{-1} - (P_0-z)^{-1}\big)\ud z\\
 &=\frac{\ui}{2\pi}\int_\gamma (P_0-z)^{-1} \big(P_0 - P^\chi \big)(P^\chi-z)^{-1} \ud z= \mathcal{O}(\eps)\,.
\end{align*}
To complete the proof, we will need to control the commutator of $\chi_2(H)$ 
with ${R^\chi(z):= \left(P^\chi-z\right)^{-1}}$.
First of all ${(1-\chi_1)\chi_2^s=0}$ for every $s>0$, so~\eqref{P chi comm2} holds for every positive power of $\chi_2$ and we can apply Lemma~\ref{lem:chi} with $T=P^\chi$ to get
\begin{equation}\label{comm Pchi}
\norm{[P^\chi,\chi_2]}_{\mathscr{L}\left(\mathscr{H},D(H)\right)}=\mathcal{O}(\eps^{N+1})\,. 
\end{equation}
Then
\begin{align}
\lVert[R&^\chi(z),\chi_2]\rVert_{\mathscr{L}\left(\mathscr{H},D(H)\right)}\notag\\
&=\norm{R^\chi(z)[P^\chi,\chi_2]R^\chi(z)}_{\mathscr{L}\left(\mathscr{H},D(H)\right)}
=\mathcal{O}(\eps^{N+1})\,.\label{comm Rchi}
\end{align}
Now since $\varrho(H)=\chi_2 (H) \varrho(H)$ we have
\begin{align*}
\lVert[H&,P_\eps]\varrho(H)\rVert_{\mathscr{L}\left(\mathscr{H}\right)}\\
&=\norm{\frac{\ui}{2\pi}\int_\gamma R^\chi(z)\left[H,P^\chi\right]R^\chi(z)\chi_2(H)\varrho(H)\ud z}\\
&=\begin{aligned}[t]
\bigg\lVert\frac{\ui}{2\pi}\int_\gamma &
R^\chi(z)\left[H,P^\chi\right]\chi_2(H)R^\chi(z)\varrho(H)\\
&+ R^\chi(z)
\underbrace{\left[H,P^\chi\right]}_{\stackrel{\eqref{P chi comm1}}{=}\mathcal{O}(\eps)}
\underbrace{\left[R^\chi(z),\chi_2(H)\right]}_{\stackrel{\eqref{comm Rchi}}{=} \mathcal{O}(\eps^{N+1})}
\varrho(H)\ud z\bigg\rVert\\                                        
\end{aligned}\\
&\leq\bigg\lVert\frac{\ui}{2\pi}\int_\gamma R^\chi(z)\underbrace{\left[H,P^\chi\right]\chi_2(H)}_{\stackrel{\eqref{P chi comm2}}{=}\mathcal{O}(\eps^{N+1})}R^\chi(z)\varrho(H)\ud z\bigg\rVert + \mathcal{O}(\eps^{N+2})\\
&= \mathcal{O}(\eps^{N+1})\,.
\end{align*}
\end{proof}
The projection $P_\eps$ has the same asymptotic expansion as $P^N$, as proved by Nenciu~\cite{Nen}.
\begin{lem}\label{lem:expansion}
 Let $P_k$ be the operators of Lemma~\ref{lem:sadiabatic}. Then for every regular cut-off $\chi\in \mathscr{C}^\infty_0\big((-\infty, \Lambda], [0,1]\big)$
\begin{equation*}
\Big\lVert \big( P_\eps - \sum_{k=0}^N \eps^k P_k\big)\chi(H) \Big\rVert_{\mathscr{L}(\mathscr{H},D(H))}=\mathcal{O}(\eps^{N+1})\,.
\end{equation*}
\end{lem}
A proof adapted to our notation can be found in~\cite[Lemma 2.25]{Lam}.

\section{The ground state band}
In this section we apply the general theory just developed to the ground state band $\lambda_0$. We begin by showing that the gap condition holds if $F$ is connected. Estimates on the size of the spectral gap in terms of geometric quantities have been derived for many special cases, mostly with $V=0$, see Schoen and Yau~\cite{SY} for a discussion of such results. 
\begin{prop}\label{prop:gap}
 Let $\lambda_0:=\min \sigma(H_F)$ be the ground state band. If $F$ is connected and $M$ satisfies Condition~\ref{cond:geom}, then $\lambda_0$ has a spectral gap in the sense of Condition~\ref{cond:gap}.
\end{prop}
\begin{proof}
We argue that the absence of a spectral gap leads to a contradiction to the fact that the ground state of a real Schrödinger operator on a connected, compact manifold is a simple eigenvalue. Let $\lambda_1:=\min \big(\sigma(H_F)\setminus \lambda_0 \big)$. If $\inf_{x\in B}\lambda_1-\lambda_0$ is not larger than zero, then clearly there exists a sequence $(x_k)_{k\in \N}$ in $B$ with $\lim_{k\to \infty} \lambda_1(x_k) -\lambda_0(x_k)=0$. Now for every $k$ take an open set $U_{\nu(k)}\in \mathfrak{U}$ containing $x_k$ and let $g_k:= (\Phi_{\nu(k)}^{-1})^*g_{F_{x_k}}$. Because of the bounds on $(\Phi_{\nu(k)}^{-1})^*$ that are required by Condition~\ref{cond:geom}, for any $m\in\N$ the sequence $(g_k)_{k\in \N}$ is bounded in the $\mathscr{C}^{m+1}$-norm on $\Gamma(T^*F\otimes T^*F)$ with respect to $g_0$. Thus by the Arzelà-Ascoli theorem there is a subsequence converging to a symmetric bilinear form $g_\infty$ of $\mathscr{C}^m$-regularity and by repeated extraction of subsequences and a diagonal argument 
$g_\infty$ is a smooth tensor. Because of the bounds on the inverse $\Phi_{\nu(k)}^*$, the sequence of metrics is also 
positive definite in a uniform way and $g_\infty$ is a Riemannian metric. Further extraction of subsequences gives convergence of $V(x_k)$ to a potential $V_\infty$. Now it was shown by Bando and Urakawa that the eigenvalues depend continuously on the metric and the potential (see~\cite{BaUr}, the proof is stated for manifolds without boundary but carries over to the Dirichlet Laplacian and Schrödinger operators because the eigenvalues are determined by a max-min principle in a similar way). This means that the sequences $\lambda_1(x_k)$ and $\lambda_0(x_k)$ converge to the two smallest eigenvalues of the operator $H_\infty:=-\Delta_{g_\infty} + V_\infty$ on $F$. But this is impossible because the smallest eigenvalue of $H_\infty$ is simple since $F$ is connected. Thus a positive lower bound for $\lambda_1-\lambda_0$ must exist.
 
 The continuous dependence of the eigenvalues on the metric and potential now shows that these eigenvalues may be separated by continuous functions, so Condition~\ref{cond:gap} is satisfied.
\end{proof}
From now on we always assume $F$ to be connected. 
In general it is convenient to express the effective and adiabatic operators using the induced connection on~$\mathcal{E}$
\begin{equation*}
 \nabla^B_X \psi := P_0X^* \psi\,,
\end{equation*}
which is usually called the Berry connection. Due to the special properties of the ground state eigenfunction, this connection and the adiabatic operator for the band $\lambda_0$ can be calculated rather explicitly. In fact it is always possible to chose an eigenfunction $\phi_0(x,\cdot) \in \ker(H_F(x)-\lambda_0(x))$ that is real valued, positive and normalised. Since this choice is unique, it provides a trivialisation of $\mathcal{E}$ and an isomorphism $L^2(\mathcal{E})\cong L^2(B)$. Expressing $\nabla^B$ in this trivialisation defines a complex-valued one-form $\omega^B$ by
\begin{equation*}
 \nabla^B_X\phi_0\psi=: \phi_0\big(X + \omega^B(X)\big)\psi\,.
\end{equation*}
Since $\phi_0$ is real, the imaginary part of $\omega^B$ vanishes. The real part can be calculated by
\begin{align*}
 2\omega^B(X)&=\omega^B(X) + \overline{\omega^B(X)}\\
&=  \int_{F_x} \big(\phi_0X^*\phi_0\big) +  \big((X^*\phi_0)\phi_0\big)\,\mathrm{vol}_{g_{F_x}}\\
&= - \int_{F_x} \abs{\phi_0}^2\mathcal{L}_{X^*}\mathrm{vol}_{F_x}\\
&= - \int_{F_x} \abs{\phi_0}^2g_B(X, \pi_*\eta)\mathrm{vol}_{F_x}
\end{align*}
and equals the mean curvature vector $\eta$ of the fibres, averaged by the eigenfunctions. Note that a non-zero real part means that $\nabla^B$ is not a metric connection and that for $\partial M =\varnothing $ this is explicitly given by
\begin{equation*}
 \omega^B\stackrel{\partial M=\varnothing}{=}-\tfrac12 \ud \big(\log\mathrm{Vol}(F_x)\big)\,.
\end{equation*}
Using these formulas an elementary calculation yields (see~\cite[Chapter 3]{Lam})
\begin{equation*}
 H_\mathrm{a}= - \eps^2 \Delta_{g_B} + \lambda_0 + \eps P_0 H_1 P_0 + \eps^2 V_\mathrm{a}\,,
\end{equation*}
with the adiabatic potential (in some contexts referred to as Born-Huang potential)
\begin{equation*}
 V_\mathrm{a}:= - \tfrac 12 \tr_{g_B}\big((\nabla_{\cdot}\omega^B)(\cdot)\big) + \int_{F_x} \pi^*g_B(\grad \phi_0, \grad \phi_0)\, \mathrm{vol}_{F_x}\,.
\end{equation*}
If the boundary is empty this evaluates to
\begin{equation}\label{eq:Va}
 V_\mathrm{a}\stackrel{\partial M=\varnothing}{=} \tfrac{1}{2} \Delta (\log \mathrm{Vol}(F_x)) + \tfrac{1}{4} \abs{\ud\log \mathrm{Vol}(F_x)}^2_{g_B} \,.
\end{equation}
Explicit formulas for this operator are derived in~\cite{HLT} for different generalisations of the waveguides and layers studied in~\cite{BMT, BGRSweakly, CEKtop, CDFK,deO, deOVe, DEbound, FrSo, Ga, GoJa, Gru, KoVu, LiLu2, Wi}. In these situations, the operator $H_1$ arises naturally from the induced metric of the tube, which is a Riemannian submersion to leading order but not exactly. The corrections to the metric concern only the horizontal directions, so $H_1$ is a horizontal differential operator of second order and the condition $-\eps^2\Delta_h + \eps H_1 \geq -C\eps$ of Proposition~\ref{prop:ground} is satisfied.
\subsection{Proof of Proposition~\ref{prop:ground}}
We now prove Proposition~\ref{prop:ground}, which states that the effective operator for the ground state band $\lambda_0$ (constructed for $N\in \N$ and $\Lambda>0$), is almost unitarily equivalent to $H$ at energies below $\Lambda_1:= \inf_{x_\in B} \lambda_1(x)$. That is, for a regular cut-off $\chi\in \mathscr{C}^\infty_0\big((-\infty, \Lambda_1), [0,1]\big)$ we need to show that
\begin{equation*}
\norm{H \chi(H) - U_\eps H_\mathrm{eff}  \chi(H_\mathrm{eff})U_\eps^*}=\mathcal{O}(\eps^{N+1})\,.
\end{equation*}
\begin{proof}
 Let $P_\eps$ be the projection of Theorem~\ref{thm:proj} for $N\in \N$ and $\Lambda>0$. We then have
 \begin{align}
  H\chi(H)&=
  \left(P_\eps H P_\eps + P_\eps^\perp H P_\eps^\perp + (1-2P_\eps)[H,P_\eps]\right) \chi(H)\notag\\
  &= U_\eps H_\eff U^*_\eps P_\eps \chi(H) + P_\eps^\perp H P_\eps^\perp \chi(H) + \mathcal{O}(\eps^{N+1})\label{eq:H_chi_ground}\,.
 \end{align}
Now, we use Lemma~\ref{lem:chi} with $T=P_\eps^\perp$ to get
\begin{equation*}
\norm{P_\eps^\perp\chi(H)-P_\eps^\perp\chi(P_\eps^\perp H P_\eps^\perp)}_{\mathscr{L}(\mathscr{H},D(H))}=\mathcal{O}(\eps^{N+1})\,.
\end{equation*}
The lower bound on $-\eps^2\Delta_h + \eps H_1$ then implies that $\supp \chi \cap \sigma(P_\eps^\perp H P_\eps^\perp)=\varnothing$, as observed in~\eqref{eq:lowerbound}. To make that observation rigorous, first note that the graph norms of $H$ and $H_{\rm diag}=P_\eps H P_\eps + P_\eps^\perp H P_\eps^\perp$ are equivalent, with constants independent of $\eps$. Then
\begin{equation*}
 \norm{P_\eps^\perp (P_0^\perp HP_0^\perp - H)P_\eps^\perp}_{\mathscr{L}(D(H_{\rm diag}), \mathscr{H})}=\mathcal{O}(\eps)\,,
\end{equation*}
so indeed (cf.~\cite[Theorem x.12]{ReSi2})
\begin{equation*}
 P_\eps^\perp H P_\eps^\perp \geq P_0^\perp HP_0^\perp - \mathcal{O}(\eps)\geq \Lambda_1 -\mathcal{O}(\eps)\,,
\end{equation*}
as an operator on $P_\eps^\perp D(H)\subset P_\eps^\perp \mathscr{H}$.
Hence $P_\eps^\perp \chi(P_\eps^\perp H P_\eps^\perp)=0$, for $\eps$ small enough, and $H P_\eps^\perp \chi(H)=\mathcal{O}(\eps^{N+1})$ for the second term in~\eqref{eq:H_chi_ground}.

It remains to prove that the first term is close to the desired one. This follows easily by another use of Lemma~\ref{lem:chi}, with $T=P_\eps$, giving
\begin{equation*}
 \norm{U^*_\eps P_\eps \chi(H)- P_0 \chi(H_\mathrm{eff}) U_\eps^*}_{\mathscr{L}(\mathscr{H},D(H))}=\mathcal{O}(\eps^{N+1})\,,
\end{equation*}
because functional calculus commutes with unitaries.
\end{proof}

%
%
%
%
\subsection{Refined asymptotics for small energies}\label{sect:low}
In this section we take a closer look at the asymptotics for small energies. We already know that $H$ may be represented by $H_\mathrm{eff}$ using the unitary transformation $U_\eps$, with arbitrary (polynomial) precision. 
We will now show that the adiabatic operator determines the spectrum of $H$ with higher precision than usual in the low energy regime. In particular, the approximation is good enough to make the adiabatic potential, which is of order $\eps^2$, relevant.
Although $H_\mathrm{a}$ depends only on $P_0$ and not the refined projections $P_\eps$, knowledge of the precise form on these projections is crucial for our proof of this approximation. 
There are two reasons for this, the first being that we need to know the terms of $H_\mathrm{a} - H_\mathrm{eff}$ rather explicitly in order to see how their size depends on the energy scale.
The second reason is that $H_\mathrm{eff}$ is not close to $H$ but only (almost) unitarily-equivalent. For the spectral problem this can be understood by taking the eigenfunctions of $H_\mathrm{eff}$ unitarily transformed with $U_\eps$ as trial functions for $H$, which does not change the eigenvalue but requires the existence of $U_\eps$.

Here we will only consider connected fibres and an operator $H_1$ of a special form, that is relevant to the applications in~\cite{HLT} and~\cite{LaNod}.
By small energies we mean energies whose distance to
\begin{equation*}                                                                                                                                                                                                                                                                                                                                                                                                                                                               \Lambda_0:=\inf_{x\in B} \min \sigma(H_F)                                                                                                                                                                                                                                                                                                                                                                                                                                                        \end{equation*}
is of order $\eps^\alpha$, with $0<\alpha \leq 2$.
It is then convenient to set
\begin{equation}\label{eq:Hground}
 H:=-\Delta_{g_\eps} + V +  \eps H_1 - \Lambda_0\,,
\end{equation}
and
\begin{equation*}
 H_F:=-\Delta_{F} +V - \Lambda_0\,.
\end{equation*}\label{not:H_F0}
For the following we fix the projection $P_\eps$ and the unitary $U_\eps$ constructed for $\lambda_0$, given $\Lambda$ and $N\geq 3$. Analysing energies of order $\eps^\alpha$ amounts to studying $H$ only on the image of 
\begin{equation}\label{eq:rho alpha}
 \varrho_\alpha(H):= 1_{(-\infty, \eps^\alpha C]}(H)\,,
\end{equation}
for some constant $C>0$.
Equivalently one may rescale the original problem by $\eps^{-\alpha}$ and consider bounded energies. The most relevant energy scales are
\begin{itemize}
 \item $\alpha=1$: If $\lambda_0$ has a unique non-degenerate minimum on $B$, the smallest eigenvalues of $-\eps^2\Delta_B + \lambda_0$ behave like those of a $d$-dimensional harmonic oscillator. 
 In particular their difference is of order $\eps$. We will show that this implies existence of eigenvalues of $H$ with the same behaviour, that are approximated by those of $H_\mathrm{a}$ up to order $\eps^{3}$.
 \item $\alpha=2$: Assume $\lambda_0 \equiv 0$, for example because $\partial M=\varnothing$ or the fibres are isometric. Then (if $H_1=0$) $H_\mathrm{a}=\eps^2(-\Delta_{g_B} + V_\mathrm{a})$ in the trivialisation by $\phi_0$, so the typical energy scale of this operator is $\eps^2$. We will show that, also for $H_1\neq 0$, small eigenvalues of $H_\mathrm{a}$ approximate those of $H$ up to $\eps^4$ and vice versa (Proposition~\ref{prop:low eigen}). More generally, the spectra coincide up to order $\eps^3$ (Proposition~\ref{prop:low spec}).
\end{itemize}
The reason why one should expect the adiabatic approximation to be better on these $\eps$-dependent energy scales is that the corrections derived in Section~\ref{sect:proj} are given by differential operators. More precisely, $P_\eps - P_0 \approx \eps P_1$ (see Lemma~\ref{lem:expansion}) with 
\begin{equation*}
 P_1P_0=-R_F(\lambda_0)[H,P_0]P_0=R_F(\lambda_0)\big(\eps[\Delta_h, P_0] - [H_1, P_0]\big)P_0\,.
\end{equation*}
By commuting derivatives to the right, the first term here can be written as a the sum of a potential of order $\eps$ and an operator with an horizontal derivative acting to the right. Now such a derivative is of order one when $\eps^2\Delta_h=\mathcal{O}(1)$, but we expect it to be of order $\eps^{\alpha/2}$ when $\eps\Delta_h=\mathcal{O}(\eps^\alpha)$, which is the case on the image of $\varrho_\alpha$. Hence we expect $\eps P_1 \varrho_\alpha(H)$ to be of order $\eps^{1+\alpha/2}$, at least if $H_1$ also consists of horizontal differential operators of non-zero order. In this case the adiabatic approximation should be better by at least a factor of $\eps^{\alpha/2}$ compared to the general case. Precisely the assumptions we make on $H_1$ are:
\begin{cond}\label{cond:low}
The operator $H_1$ has the form
\begin{equation*}
H_1\psi =-\eps^2 \divg_g S_\eps(\ud \psi, \cdot)  + \eps V_\eps\psi\,,
\end{equation*}
with $S_\eps\in \Gamma_b(\pi^*TB \otimes \pi^*TB)$ and $V_\eps \in \mathscr{C}_b^\infty$ bounded uniformly in $\eps$. 
\end{cond}
Note that such an operator always satisfies the conditions of Proposition~\ref{prop:ground}. 
With this definition we can make the heuristic discussion above precise.
\begin{lem}\label{lem:low}
 Let $0<\alpha \leq 2$, $A\in \lbrace H, H_\mathrm{a}, H_\mathrm{eff} \rbrace$, $k\in \N$ and denote by $D^k_\alpha(A)$ the domain of $\eps^{-k\alpha}A^k$ with the graph-norm. If $H_1$ satisfies condition \ref{cond:low} and $\varrho_\alpha$ is given by~\eqref{eq:rho alpha} the following hold true:
\begin{enumerate}
\item $\norm{H_1 P_0}_{\mathscr{L}(D^2_\alpha(A), D(H))}=\mathcal{O}(\eps^{\alpha/2})$,
\item $\norm{[-\eps^2\Delta_h,P_0]P_0 \varrho_\alpha(A)}_{\mathscr{L}(\mathscr{H})}
=\mathcal{O}(\eps^{1+\alpha/2})$,
\item $\norm{(P_\eps-P_0)P_0\varrho_\alpha(A)}_{\mathscr{L}(\mathscr{H},D(H))}=\mathcal{O}(\eps^{1+\alpha/2})$.
\end{enumerate}
\end{lem}
\begin{proof}
 We only sketch the proof here since it uses only standard techniques, a detailed derivation for $A\in \lbrace H, H_\mathrm{a}\rbrace$ can be found in~\cite{Lam}. The statements for $A=H_\mathrm{eff}$ follow from those for $H_\mathrm{a}$ and the fact that $H_\mathrm{eff}=H_\mathrm{a} + \mathcal{O}(\eps^2)$ (c.f.~\eqref{eq:Hsa}).
 
 The basis is to prove that for every $X \in \Gamma_b(TB)$ we have
 \begin{equation*}
\norm{\eps 
P_0 X^*}_{\mathscr{L}(D^2_\alpha(A), D(H))}=\mathcal{O}(\eps^{\alpha/2})\,,  
 \end{equation*}
which follows by showing elliptic estimates while keeping track of $\eps$ 
(see~\cite[Appendix C]{Lam}). The first statement then follows immediately from 
Condition~\ref{cond:low}. The second claim follows by writing
 \begin{align*}
  [\eps^2\Delta_h, P_0]P_0&=\eps^2\tr_{NF}[\nabla^2, P_0]P_0 - \eps^2[\eta, P_0]P_0\\
  &=\begin{aligned}[t]
  &2\eps^2\tr_{g_B}\Big([(\cdot)^* - g_B(\pi_*\eta, \cdot),P_0 ] \nabla^B_{\displaystyle\cdot}\Big)\\
  &+\underbrace{P_0^\perp\Big(\eps^2\tr_{g_B}\Big( \big[(\cdot)^*, [(\cdot)^*, P_0]\big] - [(\nabla_{\displaystyle\cdot} \cdot)^*, P_0]\Big) - [\eps^2\eta, P_0]\Big)P_0}_{\in \mathcal{A}^{0,2}}
    \end{aligned}
    \end{align*}
and applying elliptic estimates for $\nabla^B$ to the first term. The last claim follows from the second one and Lemma~\ref{lem:expansion}.
\end{proof}
We now prove several propositions that give the precise results on the low energy regime (note that we have subtracted $\Lambda_0$ from $H$). In particular they imply the statements made in Section~\ref{sect:thm}. Proposition~\ref{prop:low spec} contains general information on the spectrum of $H$ in relation to that of $H_\mathrm{a}$, Proposition~\ref{prop:low eigen} strengthens this statement for small eigenvalues of these operators and Proposition~\ref{prop:eigenfct} shows approximation for the eigenfunctions of simple eigenvalues.
This last proposition is the starting point for the investigation of nodal sets, conducted in~\cite{LaNod} and~\cite[Chapter 3]{Lam}.

\begin{prop}\label{prop:low spec}
Let $0<\alpha \leq 2$ and assume $H_1$ satisfies Condition \ref{cond:low}. Then for every $C>0$
\begin{equation*}
\dist\big(\sigma(H)\cap (-\infty, C\eps^\alpha],\sigma(H_\mathrm{a})\cap (-\infty,  C\eps^\alpha]\big)=\mathcal{O}(\eps^{2+\alpha/2})\,.
\end{equation*}
\end{prop}
\begin{proof}
Since we have chosen $N\geq 3$, Proposition~\ref{prop:ground} tells us that 
\begin{equation*}
  \dist\big(\sigma(H)\cap (-\infty, C\eps^\alpha],\sigma(H_\mathrm{eff})\cap (-\infty,  C\eps^\alpha]\big)=\mathcal{O}(\eps^{4})
\end{equation*}
 if $C\eps^\alpha<\Lambda_1$. Thus we only need to show closeness of the spectra of $H_\mathrm{a}$ and $H_\mathrm{eff}$. We can do this by expanding $H_\mathrm{eff} - H_\mathrm{a}$ on the image of $\varrho_\alpha(A)$ with $A\in \lbrace H_\mathrm{a}, H_\mathrm{eff} \rbrace$, which amounts to expanding $P_\eps U_\eps -P_0$. First we may note that $(P_0-P_\eps)^2$ commutes with both $P_0$ and $P_\eps$, since e.g. $P_0(P_\eps -P_0)^2= - P_0(P_\eps- P_0)P_0$. Thus by Definition~\eqref{eq:Udef} we have
\begin{equation*}
P_\eps U_\eps=P_\eps P_0 + \tfrac12 P_\eps P_0(P_\eps-P_0)^2 + \mathcal{O}(\eps^4)\,, 
\end{equation*}
 where the error is estimated in $\mathscr{L}(D(H))$. Inserting $P_\eps=P_0 + (P_\eps-P_0)$ gives
 \begin{equation*}
 P_\eps U_\eps = P_0 + \underbrace{P_0^\perp(P_\eps - P_0)P_0}_{=:\eps U_1} 
\underbrace{-\tfrac12 P_0 (P_\eps-P_0)^2P_0}_{=:\eps^2 U_2}
+\mathcal{O}(\eps^3)\,.
\end{equation*}
Thus we have
\begin{align}
\hspace*{-.3cm}H_\mathrm{eff} - H_\mathrm{a}
&= \eps\left(U^*_1 H P_0 + P_0 H U_1 \right)
+ \eps^2\left( U_1^* H U_1 + P_0 H U_2 + U_2 H P_0\right)+ \mathcal{O}(\eps^3)
\notag\\ 
&=\begin{aligned}[t]
   &\eps P_0(P_\eps -P_0)[H, P_0]P_0 + \eps^2 U_2 P_0 H P_0\\
   &+\eps P_0[P_0, H] (P_\eps -P_0)P_0 + \eps^2\left( U_1^* H U_1 + P_0 H U_2 \right)
   + \mathcal{O}(\eps^3)\,,
\end{aligned}\label{eq:Hsa}
\end{align}
 with an error in $\mathscr{L}(D(H_\mathrm{eff}), \mathscr{H})$. Now on the image of $\varrho_\alpha(A)$, the terms of the first line are of order $\eps^{2+\alpha/2}$ by part \textit{1.} and \textit{2.} of Lemma~\ref{lem:low} and the fact that $H_\mathrm{a} \varrho_\alpha(A)=\mathcal{O}(\eps^\alpha)$. The terms of the second line are of order $\eps^{2+\alpha/2}$ by part \textit{3.} of Lemma~\ref{lem:low}.
 This shows 
 \begin{equation*}
 \dist\big(\sigma(H_\mathrm{eff})\cap (-\infty, C\eps^\alpha],\sigma(H_\mathrm{a})\cap (-\infty,  C\eps^\alpha]\big)=\mathcal{O}(\eps^{2+\alpha/2})
\end{equation*}
by the Weyl sequence argument of Corollary~\ref{cor:spectrum}, and concludes the proof.
\end{proof}
\begin{prop}\label{prop:low eigen}
 Let $0<\alpha \leq 2$ and assume $H_1$ satisfies condition~\ref{cond:low}. If there are positive constants $C$, $\delta$ and $\eps_0$ such that
$\sigma(H_\mathrm{a}) \cap \big(-\infty, C\eps^\alpha\big)$
consists of $K+1$ eigenvalues $\mu_0\leq \dots \leq\mu_{K}$ (repeated according to multiplicity) and $\mathrm{rank}\big(1_{(-\infty, (C+\delta)\eps^\alpha)}(H_\mathrm{a})\big)<\infty$ for all $\eps<\eps_0$, 
then $H$ has $K+1$ eigenvalues $\lambda_0\leq  \dots \leq \lambda_{K}$ below the essential spectrum and
\begin{equation*}
\abs{\lambda_j -\mu_j}=\mathcal{O}(\eps^{2+\alpha}) 
\end{equation*}
for all $j\leq K$.
\end{prop}
\begin{proof}
Because of the unitary equivalence up to $\mathcal{O}(\eps^4)$ shown in Proposition~\ref{prop:ground} it is enough to prove the claim for $H_\mathrm{eff}$ instead of $H$. It follows from~\eqref{eq:Hsa} and Lemma~\ref{lem:low} that for normalised $\psi$ in the image of $\varrho_\alpha(A)$ with $A\in \lbrace H_\mathrm{a}, H_\mathrm{eff} \rbrace$
\begin{equation}\label{eq:Hsa quad}
 \abs{\langle \psi, (H_\mathrm{a} - H_\mathrm{eff}) \psi \rangle}=\mathcal{O}(\eps^{2+\alpha})\,,
\end{equation}
because in the quadratic form $\varrho_\alpha(A)$ acts both from left and right onto ${H_\mathrm{a}-H_\mathrm{eff}}$.
Now if $\mathrm{rank}\big(1_{(-\infty, (C+\delta/2)\eps^\alpha)}(H_\mathrm{eff})\big)$ were infinite, then~\eqref{eq:Hsa quad} would imply that the subspace where $\langle \psi, H_\mathrm{a} \psi\rangle\leq (C+\delta)\eps^\alpha \norm{\psi}^2$ had infinite dimension, in contradiction to the hypothesis. Consequently
$\sigma(H_\mathrm{eff})\cap (-\infty, (C+\delta/2)\eps^\alpha)$ consists of finitely degenerate eigenvalues $\tilde \lambda_0\leq  \cdots$.
The eigenvalues of $H_\mathrm{a}$, and also those of $H_\mathrm{eff}$, are characterised by the min--max principle
\begin{equation*}
\mu_j= \min_{W_j} \max \lbrace \langle \psi, H_\mathrm{a} \psi\rangle : \psi\in W_j, \norm{\psi}_{L^2(\mathcal{E})}=1 \rbrace\,,
\end{equation*}
where $W_j$ runs over the $(j+1)$-dimensional subspaces of $D(H_\mathrm{a})=D(H_\mathrm{eff})$.
Choosing $W_j \subset \bigoplus_{k=0}^{j} \ker(H_\mathrm{a} - \mu_k)$ gives
\begin{align*}
\tilde \lambda_j
&\stackrel{\eqref{eq:Hsa quad}}{\leq}
\max \lbrace \langle \psi, H_\mathrm{a} \psi\rangle : \psi\in W_j, \norm{\psi}_{L^2(\mathcal{E})}=1 \rbrace + \mathcal{O}(\eps^{2+\alpha})\\
&\stackrel{\hphantom{(19)}}{\leq}\mu_j  + \mathcal{O}(\eps^{2+\alpha})\,.
\end{align*}
This shows that $H_\mathrm{eff}$ has $K+1$ eigenvalues below $(C + \delta/2)\eps^\alpha$. We can thus repeat the argument with reversed roles of $H_\mathrm{a}$ and $H_\mathrm{eff}$ to obtain
\begin{equation*}
 \mu_j\leq \tilde \lambda_j + \mathcal{O}(\eps^{2+\alpha})\,,
\end{equation*}
which proves the claim.
\end{proof}
\begin{prop}\label{prop:eigenfct}
Assume the conditions of Proposition~\ref{prop:low eigen} are satisfied and additionally that $\mu \in \lbrace \mu_0, \dots , \mu_K \rbrace$ is a simple eigenvalue for which there exists $C_\mu>0$ such that $\dist(\mu,\sigma(H_\mathrm{a})\setminus \lbrace \mu \rbrace) \geq C_\mu \eps^\alpha$. Then the eigenvalue $\lambda \in \sigma(H)$ corresponding to $\mu$ is simple and $\dist(\lambda,\sigma(H)\setminus \lbrace \lambda \rbrace) \geq C_\lambda \eps^\alpha>0$.
Moreover if $\psi \in \ker(H_\mathrm{a}-\mu)$ is normalised and $P_\lambda$ denotes the orthogonal projection to $\ker(H-\lambda)$ then 
\begin{equation*}
 \norm{(1-P_\lambda)\psi}_{D(H)}=\mathcal{O}(\eps^{\beta})
\end{equation*}
with $\beta=\min\lbrace 1+ \tfrac12 \alpha, 2- \tfrac12 \alpha \rbrace$ and
\begin{equation*}
\norm{(1-P_\lambda)\psi}_{W^1_{\eps=1}}=\mathcal{O}(\eps^{\alpha/2})\,.
\end{equation*}
\end{prop}
\begin{proof}
 The statement on the eigenvalue $\lambda$ follows directly from Proposition~\ref{prop:low eigen}. Let $P_{\tilde \lambda}$ denote the projection to the eigenspace of $H_\mathrm{eff}$ corresponding to $\mu$. Then since $\psi=P_0 \varrho_\alpha(H_\mathrm{a}) \psi$, Equation~\eqref{eq:Hsa} and Proposition~\ref{prop:low eigen} imply
 \begin{equation*}
  \norm{(1-P_{\tilde \lambda})\psi}_{D(H)} \leq
  \underbrace{\norm{(H_\mathrm{eff} - \tilde \lambda)^{-1}(1-P_{\tilde \lambda})\psi}_{\mathscr{L}(\mathscr{H},D(H))}}_{=\mathcal{O}(\eps^{-\alpha})}
 \underbrace{\norm{(H_\mathrm{eff}- \tilde \lambda) \psi}_\mathscr{H}}_{=\mathcal{O}(\eps^{2+\alpha/2})}\,.
 \end{equation*}
Now by Proposition~\ref{prop:ground}, $U_\eps^* P_{\lambda}U_\eps=P_{\tilde \lambda} + \mathcal{O}(\eps^4)$, so
\begin{equation*}
\norm{(1-P_\lambda)U_\eps \psi}_{D(H)}=\mathcal{O}(\eps^{2-\alpha/2})\,.
\end{equation*}
We also have
\begin{equation*}
 \norm{(U_\eps -1 )\psi}_{D(H)}=\norm{P_0^\perp(P_\eps - P_0)\varrho_\alpha(H_\mathrm{a})\psi}_{D(H)} + \mathcal{O}(\eps^2)\,,
\end{equation*}
which implies $\norm{(1-P_\lambda)\psi}_{D(H)}=\mathcal{O}(\eps^\beta)$ by Lemma~\ref{lem:low}.

For the second estimate, note that $\phi:=(1-P_\lambda)\psi$ satisfies Dirichlet conditions, so we can use the elliptic estimate Theorem~\ref{thm:ellipt} to obtain
\begin{align*}
&C^{-1}\norm{\phi}^2_{W^1_{\eps=1}}
\leq \norm{\phi}^2_\mathscr{H} + \langle \phi, (-\Delta_F - \Delta_h) \phi\rangle\\
&\leq  \left(1+\Lambda_0 + \norm{V}_\infty\right)\norm{\phi}^2 +
\eps^{-2}\langle \phi,\underbrace{(-\Delta_F + V - \Lambda_0)}_{\geq 0} \phi\rangle
+ \langle \phi, - \Delta_h \phi\rangle\\
&\leq \left(1+\Lambda_0+ \norm{V}_\infty + \eps^{-2}\abs{\lambda}\right)\norm{\phi}^2 + \eps^{-2}\abs{\left\langle \phi, (H - \lambda)\phi\right\rangle}
+ \eps^{-1}\abs{\left\langle \phi, H_1\phi\right\rangle}\,.
\end{align*}
Now since $P_\lambda (H-\lambda)=0$ and $\psi=P_0\psi$:
\begin{align*}
 \eps^{-2}\abs{\langle \phi, (H- \lambda)\phi \rangle}=
 \eps^{-2}\abs{\langle \psi, (H- \lambda)\psi \rangle}
 =\eps^{-2}\abs{\langle \psi, (\mu- \lambda)\psi \rangle}=\mathcal{O}(\eps^\alpha)\,,
\end{align*}
while $\eps^{-2}\abs{\lambda}\norm{\phi}_\mathscr{H}=\mathcal{O}(\eps^{2\beta+\alpha-2})=\mathcal{O}(\eps^\alpha)$. By Condition~\ref{cond:low}, the term containing $H_1$ can be bounded by 
\begin{equation*}
\eps^{-1}\abs{\left\langle \phi, H_1\phi\right\rangle} \leq \norm{V_\eps}_\infty \norm{\phi}^2_\mathscr{H} + \eps \norm{S_\eps}_\infty \norm{\phi}^2_{W^1_{\eps=1}}
\end{equation*}
and we conclude that
\begin{equation*}
 (C^{-1} -\eps \norm{S_\eps}_\infty)\norm{\phi}^2_{W^1_{\eps=1}} =\mathcal{O}(\eps^\alpha)\,,
\end{equation*}
which proves the claim.
\end{proof}

\end{document}